\documentclass[12pt, onecolumn, draftcls]{IEEEtran}
\normalsize
\usepackage{amsmath,amssymb,amsthm}
\usepackage{xifthen}
\usepackage{mathtools}
\usepackage{enumerate}
\usepackage{microtype}
\usepackage{xspace}
\usepackage{bm}
\usepackage[T1]{fontenc}
\usepackage{fancyhdr}
\usepackage{lastpage}
\usepackage{bbm}

\newcommand{\mbs}[1]{\pmb{#1}}
\newcommand{\vect}[1]{{\lowercase{\mbs{#1}}}}
\newcommand{\mat}[1]{{\uppercase{\mbs{#1}}}}

\newcommand{\T}{{\scriptscriptstyle\mathsf{T}}}
\renewcommand{\H}{{\scriptscriptstyle\mathsf{H}}}

\renewcommand{\Re}[1][]{\ifthenelse{\isempty{#1}}{\operatorname{Re}}{\operatorname{Re}\left(#1\right)}}
\renewcommand{\Im}[1][]{\ifthenelse{\isempty{#1}}{\operatorname{Im}}{\operatorname{Im}\left(#1\right)}}

\newcommand{\hv}{\vect{h}}

\newcommand{\pv}{\vect{p}}

\newcommand{\rv}{\vect{r}}

\newcommand{\uv}{\vect{u}}

\newcommand{\wv}{\vect{w}}
\newcommand{\xv}{\vect{x}}

\newcommand{\Am}{\mat{a}}

\newcommand{\Hm}{\mat{h}}

\newcommand{\Mm}{\mat{M}}

\newcommand{\Qm}{\mat{q}}
\newcommand{\Rm}{\mat{r}}

\newcommand{\Um}{\mat{u}}

\newcommand{\Kc}{{\mathcal K}}

\newcommand{\Pc}{{\mathcal P}}

\newcommand{\Id}{\mat{\mathrm{I}}}

\newcommand{\CN}[1][]{\ifthenelse{\isempty{#1}}{\mathcal{N}_{\mathbb{C}}}{\mathcal{N}_{\mathbb{C}}\left(#1\right)}}

\renewcommand{\P}[1][]{\ifthenelse{\isempty{#1}}{\mathbb{P}}{\mathbb{P}\left(#1\right)}}
\newcommand{\E}[1][]{\ifthenelse{\isempty{#1}}{\mathbb{E}}{\mathbb{E}\left[#1\right]}}
\newcommand{\I}[1][]{\ifthenelse{\isempty{#1}}{\mathbb{I}}{\mathbb{I}\left\{#1\right\}}}
\renewcommand{\det}[1][]{\ifthenelse{\isempty{#1}}{\mathrm{det}}{\text{det}\left(#1\right)}}
\newcommand{\trace}[1][]{\ifthenelse{\isempty{#1}}{\mathrm{tr}}{\text{tr}\left(#1\right)}}
\newcommand{\rank}[1][]{\ifthenelse{\isempty{#1}}{\mathrm{rank}}{\text{rank}\left(#1\right)}}
\newcommand{\diag}[1][]{\ifthenelse{\isempty{#1}}{\mathrm{diag}}{\text{diag}\left(#1\right)}}
\newcommand{\Cov}[1][]{\ifthenelse{\isempty{#1}}{\mathsf{Cov}}{\mathsf{Cov}\left(#1\right)}}

\newtheorem{remark}{Remark}[section]
\newtheorem{definition}{Definition}
\newtheorem{example}{Example}

\renewcommand{\rv}[1]{{\mathrm{#1}}}
\newcommand{\rvVec}[1]{\pmb{\mathrm{#1}}}
\newcommand{\rvMat}[1]{\pmb{\mathsf{#1}}}

\newcounter{enumi_saved}
\usepackage{answers}
\Newassociation{solution}{Solution}{solutionfile}

\AtBeginDocument{\Opensolutionfile{solutionfile}[\jobname]}
\AtEndDocument{\Closesolutionfile{solutionfile}\clearpage
}

\IfFileExists{MinionPro.sty}{
}{
}

\usepackage{booktabs}
\usepackage{amsmath}
\usepackage{subfiles}  
\allowdisplaybreaks[4]
\usepackage{algorithm}
\usepackage{algorithmic}
\usepackage{color}
\usepackage{threeparttable}
\usepackage{multirow}
\usepackage{multicol}
\usepackage{multirow}

\newcommand{\collectionK}{\boldsymbol{\mathcal{K}}}
\newcommand{\collectionS}{\boldsymbol{\mathcal{S}}}
\newtheorem{claim}{Claim}

\begin{document}
\title{Rate Splitting for Multi-Antenna Downlink: \\ Precoder Design and Practical Implementation}
\author{Zheng Li, Chencheng Ye, Ying Cui, Sheng Yang, and Shlomo Shamai~(Shitz)\thanks{Z. Li and S. Yang are with the laboratory of signals and systems at
		CentraleSup\'elec-CNRS-Universit\'e Paris-Sud, 91192,
		Gif-sur-Yvette, France. Z. Li is also with Orange Labs
		Networks, 92326, Ch\^atillon, France. C. C. Ye and Y. Cui are with Shanghai Jiao Tong University, China.
		S.~Shamai~(Shitz) is with Technion-Israel Institute of Technology, Haifa, Israel. Email:\{zheng.li,
		sheng.yang\}@centralesupelec.fr, \{yechencheng, cuiying\}@sjtu.edu.cn,  sshlomo@ee.technion.ac.il.}
	 \thanks{The work of S. Shamai has been supported by the
		European Union's Horizon 2020 Research And Innovation Programme,
		grant agreement no.~694630.} }
\maketitle
\begin{abstract} 

  Rate splitting~(RS) is a potentially powerful and flexible technique
  for multi-antenna downlink transmission. In this paper, we address
  several technical challenges towards its practical implementation for
  beyond 5G systems. To this end, we focus on a single-cell system with
  a multi-antenna base station (BS) and $K$ single-antenna receivers. We consider RS in its most general form with $2^K-1$ streams, and joint decoding to fully exploit the potential of RS.
  First, we investigate the achievable rates under joint decoding and formulate the precoder design problems to maximize a general utility function, or to minimize the transmit power under pre-defined rate targets.  
  Building upon the concave-convex procedure~(CCCP), we propose
  precoder design algorithms for an \emph{arbitrary} number of users. Our proposed algorithms
  approximate the intractable non-convex problems with a 
  number of successively refined convex problems, and provably converge to
  stationary points of the original problems. Then, to reduce the decoding complexity, we
  consider the optimization of the precoder and the decoding order under successive decoding. Further, we  propose a stream
selection algorithm to reduce the number of precoded signals.
 With a reduced number of
  streams and successive decoding at the receivers, our proposed algorithm can even be
  implemented when the number of users is relatively large, whereas the
  complexity was previously considered as prohibitively high in the same
  setting. Finally, we propose a simple adaptation of our
  algorithms to account for the imperfection of the channel state
  information at the transmitter. Numerical results demonstrate 
  that the general RS scheme provides a substantial performance gain as compared to
  state-of-the-art linear precoding schemes, especially with a
  moderately large number of users. 
\end{abstract}

\newcommand{\Comment}[1]{\textbf{\color{red} #1}}

\section{Introduction}	

While the first version of the fifth generation~(5G) has been recently
deployed, many communication requirements for future applications, e.g., exceptionally high
bit rates and high energy efficiency, remain unaddressed. A plethora of
new multi-antenna~(MIMO) transmission techniques, such as cell-free
massive MIMO, hybrid beamforming, lens antenna arrays, and large
intelligent surface~(LIS), have been recently proposed for that purpose.
Nevertheless, even with these new techniques, interference is still the 
fundamental barrier towards a better performance in a wireless MIMO
network, especially at downlink \cite{bogale2016massive, zhu2017hybrid}.

Implementing MIMO downlink is challenging for several reasons. First, to
mitigate interference at the receivers' side, precoding that relies on
precise channel state information at the transmitter's side~(CSIT)
is needed. Such information is hard to obtain especially at high mobility. Second,
even with perfect CSIT, precoder design is non-trivial. The optimal
precoder that achieves the capacity region is known as dirty paper
coding~(DPC)~\cite{Caire:2003, MIMOBC2006} that is non-linear.
Implementing DPC requires vector quantization that is NP-hard.
Furthermore, it is
also well known that DPC is quite sensitive to CSIT
accuracy~\cite{yang2005impact}.  As such, in current systems, linear
precoders such as zero forcing~(ZF) are used instead.  It is well
established that ZF achieves the optimum degree of freedom (DoF) at high
signal-to-noise (SNR) ratio~\cite{yoo2006optimality,lee2007high}.  However,
the authors in \cite{li2018linearly} show that, despite its DoF
optimality, \emph{any} linear precoding scheme can be far from optimal,
since the gap between the achievable sum rate of the best linear scheme
and the sum capacity can be unbounded. Indeed, with such linear schemes,
we deal with interference in essentially two ways: 1)~the transmitter
applies \emph{interference
avoidance} by steering the signal of any user into other users' null space, and 2)~the
receivers \emph{treat interference as noise}. In other words, linear precoders are
designed such that interference power is minimized as compared to the
signal power at the receivers' side. The cost to suppress interference
can be high when the channels for some of the users are spatially aligned. 

To circumvent such limitation, the idea of rate splitting~(RS) is
basically to introduce a new option to the receivers: \emph{interference
decoding}. Specifically, each individual message is split into private
and common parts, which are respectively encoded and carried by different
signals. Each common part is decodable by~(though not necessarily
intended to) multiple receivers. Each receiver can decode and then remove
the common part before decoding the private part. In this way, part of
the interference has been removed since it is decodable, which improves
the overall performance. Originally, RS is proposed to partially
mitigate interference in the two-user interference channels~\cite{carleial1978interference, HK:1981}, in which independent messages are sent by independent transmitters to their respective receivers.
 It turns out that such a scheme achieves the capacity region of the two-user interference
channel to within one bit per channel use~(PCU)~\cite{etkin2008gaussian}. 
In \cite{yang2013degrees}, RS is applied to the multi-antenna
broadcast channel (BC) and shown to provide a strict sum DoF gain of a BC when
only imperfect CSIT is available. Recently, in \cite{li2018linearly},
the authors establish the optimality of linearly precoded RS in the
constant gap sense in the two-user MIMO BC case. 

In its most general form, the RS scheme can split each message into as
many as $2^{K-1}$ sub-messages in a $K$-user channel.  Then, the total $2^{K-1} K$ sub-messages are re-assembled into $2^K-1$ new
messages. The BS creates one directional signal for each re-assembled message, and we also refer to such signals as streams. Precoder design together with power
allocation can be done across all the sub-messages. Such great
flexibility also comes at the cost of high complexity, for both
the precoder design and the decoder implementation. The goal of this work is 
therefore to investigate the true potential of RS, and to solve some
technical challenges towards its practical implementation. Our main
contributions can be summarized into the following two items. 
\subsubsection*{\underline{Precoder design}}
We consider RS in
its most general form, i.e., with an arbitrary number $K$ of
users and an arbitrary subset of active streams. To explore
the full potential of RS, we consider joint decoding of all common
messages at each receiver. We formulate the precoder design problems to optimize the commonly used performance
metrics, such as the weighted sum rate, the worst-user rate, as well
as the transmit power~(for given target rates).  These problems are non-convex
and therefore hard to solve in general. Then, building on the
concave-convex procedure~(CCCP)~\cite{sun2017MM}, we propose algorithms to solve approximately the original
problems. Our algorithms can be proved to converge to stationary points
of the precoder design problems. 
To the best of our knowledge, this is the first work to provide
precoder design for the general RS scheme, as well as the first
work to combine RS and joint decoding. By constrast, previous works only
consider RS in reduced forms or with successive decoding. 

\subsubsection*{\underline{Practical implementation}} In addition to the
general precoder design algorithms, we also propose
further adaptations towards practical implementation of the RS scheme.
\begin{itemize}
  \item To reduce the complexity on the precoder design and the
    decoding, we propose a new stream elimination algorithm which is
    then combined with the precoder design algorithm. The remaining
    streams are such that the searching space of the decoding order is essentially reduced. With such an adaptation, the general RS scheme can be
    applied even for  a large number of users.
    Comparison among different algorithms reveals the substantial
    complexity reduction from the proposed stream selection algorithm.
  \item We propose a slight modification of the precoder design
    algorithms to account for the CSIT imperfection. Specifically,
    instead of reformulating entirely the problem, we introduce a
    regularization term in the 
    precoder design formulation according to the CSIT
    accuracy. Numerical results show that the proposed
    regularization is quite effective and can improve significantly the sum rate
    with imperfect CSIT. 
\end{itemize}

In order to validate the proposed algorithms, we have run numerical
simulations and compare the performance to existing schemes.  We show
that the general RS scheme outperforms substantially state-of-the-art
linear precoding schemes, especially with a moderately large number of
users~(e.g.,~8), both in terms of achievable rates and of total transmit
power. 

\subsubsection*{Related works}
In \cite{clerckx2019rate, joudeh2016robust,mao2018rate, yang2019optimization}, the authors
explore a structured and simplified version of RS, i.e., the 1-layer RS,
where each message is split into one private part and only one common
part. The common parts of all the messages are encoded into one
common stream that should be decoded by all the users, whereas the
private parts are unicast to the corresponding receivers.  
While the optimization problem in the 1-layer RS is simpler than the
general case, it does not take full advantage of the flexibility of the general RS.
Thus, the potential of the general RS remains unknown. 
Although the authors in \cite{mao2018rate} do mention the general RS scheme,
they only tackle the sum rate maximization problem for $K=2, 3$ users.
In fact, their method does not seem to scale with $K$, while our
formulation applies to an arbitrary number of users and an arbitrary
subset of streams. In \cite{dai2016rate}, the authors propose a hierarchical RS, which transmits one outer common message and multiple inner common messages. The outer common message can be decoded by all users while each inner common message is decodable by a subset of users. However, the authors mainly focus on the asymptotic sum rate analyses in massive MIMO systems and  the optimization of precoders of the common messages. 
Power allocation in \cite{dai2016rate} is also simplified to equal power
allocation among private messages, whereas in our work we optimize the
power allocation among all messages.  
Further, it is worth
mentioning that these works consider only successive decoding
while we consider both joint decoding and successive decoding in our
formulation to explore the full potential of the general RS. 
In \cite{li2018linearly}, the authors consider the
general $K$-user RS scheme with joint decoding, and show that MMSE
precoder achieves constant-gap capacity in the two-user MIMO BC case. They also
propose a stream elimination algorithm based on constant-gap argument. 
However, the constant-gap argument is essentially for the high SNR regime, and the important
questions of how to design precoders at finite SNR regime and how to
implement successive decoding have not been addressed.  

In~\cite{dai2016rate} and~\cite{joudeh2016partial}, RS is considered to
maximize the sum rate with imperfect CSIT. Specifically,
\cite{dai2016rate} considers a hierarchy RS and \cite{joudeh2016partial}
studies the 1-layer RS. To the best of our knowledge, no current
reference exploits the precoder design of the general RS with imperfect
CSIT. Furthermore, \cite{dai2016rate} directly assumes regularized-ZF
based on the estimated CSI as the precoder of the private messages, and
considers only the optimization of the precoders of the common messages.
In  \cite{joudeh2016partial}, the authors mainly focus on the DoF
derivation where the power goes to infinity. In terms of rate
optimization, \cite{joudeh2016partial} proposes an algorithm only based
on several samples of the channel estimation error. Such
state-of-the-art results can be far from the performance of the general
RS with imperfect CSIT. By contrast, in this paper, we propose a
simple and effective regularization to account for CSIT imperfection,
without changing the main precoder design.

The remainder of the paper is organized as follows. In Section \ref{sec:model}, we 
present the channel model, and describe the general RS strategy and its corresponding achievable rate region under joint decoding. Optimization under the general RS for joint decoding is considered in Section~\ref{sec:op_JD}. Optimization under successive decoding, stream selection and adaptation of our algorithms to imperfect CSIT scenario are presented in Section \ref{sec:JD_reduced}. Simulation results are illustrated and discussed in Section \ref{sec:simu}. Finally, the paper is concluded in Section \ref{sec:con}.

\subsubsection*{Notation}
For random quantities, we use upper case non-italic letters, e.g.,
$\rv{X}$, for scalars, upper case non-italic bold letters, e.g.,
$\rvVec{V}$, for vectors, and upper case letter with bold and sans serif
fonts, e.g., $\rvMat{M}$, for matrices. Deterministic quantities are
denoted in a rather conventional way with italic letters, e.g., a scalar
$x$, a vector $\pmb{v}$, and a matrix $\pmb{M}$. We denote $\Mm^T$,
$\Mm^H$ and $\trace(\Mm)$  as the transpose, the conjugate transpose and
the trace of a matrix $\Mm$, respectively. Sets are denoted with
calligraphic capitalized letters, e.g., $\Kc$, and $|\Kc|$ represents
the cardinality of a set $\Kc$. We also use bold calligraphic letters to
specify sets of sets, e.g., $\boldsymbol{\Kc}$. $[n]$ is the set
$\{1,\ldots,n\}$. Logarithms are to the base~$2$. We use
$\lfloor x \rfloor$ to denote the largest integer that is not larger than
$x$.

\section{System Model and Problem Formulation}\label{sec:model}
We consider a single-cell downlink communication system, where the base
station~(BS) with $M$ antennas serves $K$ single-antenna users. The
mathematical model of the communication channel during a transmission of
$T$ symbols is described as follows.
Let $\xv[t] \in \mathbb{C}^{M\times 1}$ be the transmitted signal from
the BS at time $t\in[T]$. The channel output at user $k\in[K]$ is
\begin{align}
\rv{Y}_k[t] = \hv_k^H \xv[t] +\rv{Z}_k[t],
\end{align}
where $\hv_k\in \mathbb{C}^{M\times 1}$ is the channel vector from the
BS to user~$k$; $\rv{Z}_k[t]\sim\mathcal{CN}(0,1)$ is the additive white
Gaussian noise (AWGN) with normalized variance and is independent over
time. Note that we assume that the channel remains constant during
the whole transmission.  

The goal of the BS is to transmit $K$ independent messages,
$\rv{M}_1,\ldots,\rv{M}_K$, to the $K$ users, respectively, in $T$
channel uses. Let $\rv{M}_k\in\mathcal{M}_k$, then the transmission rate
for user~$k$ is defined as $R_k = \frac{\log|\mathcal{M}_k|}{T}$ bits
per channel use~(PCU). An encoding scheme maps the $K$
messages into a sequence of $T$ symbols, say,
$\xv[1],\ldots,\xv[T]$. A rate tuple $(R_1,\ldots,R_K)$
is achievable under power constraint $P$ if there exists an encoding scheme such that the $K$
messages can be decoded at each receiver with an arbitrarily small error
probability, and that the transmitted signal satisfies
\begin{align}
\frac{1}{T}\sum_{t=1}^T \|\xv[t]\|^2 \le P, \label{eq:power_constraint}
\end{align}%
when $T\to\infty$. Unless specified otherwise, we assume that the
channel realizations $\left\{\hv_k\right\}_k$ are perfectly known at the
BS and at the receivers. 

\subsection{Linear Precoding for Unicast}

In the commonly used unicast scheme, $K$ messages are encoded separately, and a linear
	superposition of the $K$ encoded signals is sent. Specifically, the BS transmits  
\begin{align}
\rvVec{X} = \sum\nolimits_{k\in[K]} \rvVec{X}_k,\label{eq:conventional}
\end{align}
where $\rvVec{X}_k\in\mathbb{C}^{M\times 1}$ is the encoded signal for message $k$. Here, we omit the time index and adopt the
commonly used \emph{single-letter} expression where the signals $\{\xv[t]\}$ are
replaced by the random vector $\rvVec{X}$ for further analysis. We define the
covariance matrix $\Qm_k:=\E\{\rvVec{X}_k\rvVec{X}_k^H\}\succeq 0$ to specify the
precoder for signal $k$.
Accordingly, the transmit power becomes $\sum_{k\in[K]} \trace\left( \Qm_k \right)$. We call such a scheme linear precoding for unicast. The received signal at user $k$ is
\begin{align}
\rv{Y}_k=\underbrace{\hv_k^H\rvVec{X}_k}_{\text{desired signal}}+\underbrace {\sum\nolimits_{i\neq k}\hv_k^H\rvVec{X}_{i}}_{\text{interference}}+\,\rv{Z}_k,\quad \forall k\in[K].\label{eq:unicast}
\end{align}
Assuming Gaussian signaling, we can derive the achievable rate of user
$k$ as 
\begin{align}
R_k^\text{unicast}=\log\left(1+\frac{\hv_k^H\Qm_k\hv_k}{1+\sum\nolimits_{i\neq k}\hv_k^H\Qm_i\hv_k}\right),\quad \forall k\in[K],\label{eq:unicast_SINR}
\end{align}
where the fraction in the logarithm is referred to as the
signal-to-interference-and-noise-ratio~(SINR) of user $k$. Essentially,
interference is treated as noise in this scheme. One can then optimize
over the precoder, via $\{\Qm_k\}_k$, for different performance
metrics and requirements~\cite{tervo2015optimal,vucic2008robust}. 

\subsection{Linearly Precoded Rate-splitting}\label{sec:RS}

For each $k\in[K]$, let us define the following collection of subsets
of $[K]$
\begin{equation}
\collectionK^{(k)}:= \{ \mathcal{K}\subseteq [K]:\, k\in\mathcal{K}\},\quad \forall k\in[K].
\label{eq:collectionK}
\end{equation}
Namely, $\collectionK^{(k)}$ collects all $2^{K-1}$ subsets of $[K]$ that contain $k$. 
The linearly precoded rate-splitting scheme in the most general form is described as
follows. 

First, we split each message set $\mathcal{M}_k$ into sub-message
sets, such that 
\begin{equation}
\mathcal{M}_k = \prod_{\mathcal{K} \in \collectionK^{(k)}}
\mathcal{M}_{\mathcal{K}}^{(k)},
\end{equation}
where the right-hand side is the Cartesian product of $2^{K-1}$ sets. Thus, any message
$\rv{M}_k \in \mathcal{M}_k$ can be equivalently represented by a
sub-message
tuple~$\bigl(\rv{M}_{\mathcal{K}}^{(k)}:\,{\mathcal{K}\in\collectionK^{(k)}}\bigr)$
where $\rv{M}_{\mathcal{K}}^{(k)} \in
\mathcal{M}_{\mathcal{K}}^{(k)}$, $\mathcal{K} \in \collectionK^{(k)}$.
The rate $R_k$ is split into the rates
$\bigl(R_{\mathcal{K}}^{(k)}:\,{\mathcal{K}\in\collectionK^{(k)}}\bigr)$ of the
sub-messages such that 
\begin{align}
R_k = \sum_{\Kc\in\collectionK^{(k)}} {R}^{(k)}_{\mathcal{K}},\quad \forall \,k\in[K].\label{eq:rs_def}
\end{align}

Then, for each non-empty subset $\Kc \subseteq [K]$, we re-assemble the sub-messages
$\rv{M}_\Kc^{(k)}$, $k\in \Kc$, to form a message vector
\begin{equation}
{\rvVec{M}}_{\mathcal{K}} :=
\bigl({\rv{M}}^{(k)}_{\mathcal{K}}:\,k\in\mathcal{K}\bigr),\quad\forall\Kc\subseteq [K],
\end{equation}
with rate 
\begin{align}
{R}_{\mathcal{K}}=\sum_{k \in \mathcal{K}} {R}^{(k)}_{\mathcal{K}}, \quad\forall\Kc\subseteq [K].
\end{align}
That is, we re-arrange the $K 2^{K-1}$ sub-messages $\bigl({\rv{M}}^{(k)}_{\mathcal{K}}:\,k\in[K],  \mathcal{K} \in \collectionK^{(k)}\bigr)$ into $2^K-1$ message
vectors $\bigl(\rvVec{M}_{\mathcal{K}}: \Kc\subseteq [K]\bigr)$. Each of the $2^K-1$ message vectors $\rvVec{M}_\Kc$ is encoded
separately\footnote{Since each of the sub-messages in $\rvVec{M}_\Kc$ is required to be decoded by the users in $\Kc$, $\rvVec{M}_\Kc$ is effectively an integral message to the users in $\Kc$. The re-assembling at the level of information bits is thus without loss of optimality.} and carried by some signal
$\rvVec{X}_{\Kc}\in\mathbb{C}^{M\times 1}$ for which the precoder is
specified by the covariance matrix
$\Qm_{\Kc}:=\E\{\rvVec{X}_{\Kc}\rvVec{X}_{\Kc}^H\}$ satisfying
\begin{align}
\Qm_{\Kc}\succeq 0,\quad \forall\Kc\subseteq[K]. \label{eq:con_PSD} 
\end{align}

The transmitted signal is a linear superposition of all these signals
\begin{align}
\rvVec{X}=\sum\nolimits_{\Kc\subseteq[K]} \rvVec{X}_{\Kc}.
\end{align}%
At the receivers' side, each user~$k$ observes
\begin{align}
\rv{Y}_k=\underbrace{\sum_{\Kc\in\collectionK^{(k)}}\hv_k^H\rvVec{X}_{\Kc}}_{\text{desired
		signal}} + \underbrace
{\sum_{{\Kc'}\not\in{\collectionK}^{(k)}}\hv_k^H\rvVec{X}_{{\Kc'}}}_{\text{interference}}+\,\rv{Z}_k,\quad \forall k\in[K],\label{eq:RS_receiver}
\end{align}
and decodes all message vectors $\{\rvVec{M}_{\Kc}:\,
\Kc\in\collectionK^{(k)}\}$, by treating the
interference
$\sum_{{\Kc'}\not\in{\collectionK}^{(k)}}\hv_k^H\rvVec{X}_{{\Kc'}}$
as noise. Since the desired sub-messages $\{\rv{M}_{\Kc}^{(k)}:\,
\Kc\in\collectionK^{(k)}\}$ can be extracted from the message vectors
$\{\rvVec{M}_{\Kc}:\,\Kc\in\collectionK^{(k)}\}$, user~$k$ can recover the
original message $\rv{M}_k$. Note that with RS, each user decodes not
only the desired message, but also part of the messages of the other users.
Hence, the main idea of RS is to make the interfering messages partially
decodable in order to reduce the interference level.  

To analyze the achievable rate, we notice that with RS each
user is equivalent to a receiver in a multiple access channel~(MAC) in
which a number of independent messages must be decoded from the received
signal.  As in a MAC, the achievable rate of RS depends on
how the message vectors are decoded. To exploit the full potential of the general RS, we first consider joint decoding.   
In such case, each receiver $k$ jointly decodes the set of messages
$\{\rvVec{M}_{\mathcal{K}}:\,\Kc\in\collectionK^{(k)}\}$.
Then the achievable rate region of the message vectors is described by 
the following constraints \cite{li2018linearly}
\begin{align}
\sum_{\mathcal{K}\in \collectionS^{(k)}} {R}_{\mathcal{K}} \le
\log\biggl(1 + \frac{\sum_{\mathcal{K}\in \collectionS^{(k)}} \hv_k^\H
	\Qm_{\mathcal{K}} \hv_k}{
	1 + \sum_{\mathcal{K}'\not\in {\collectionK}^{(k)}}
	\hv_k^\H \Qm_{\mathcal{K}'} \hv_k}
\biggr),\quad\forall\, k\in[K], \collectionS^{(k)} \subseteq
\collectionK^{(k)}.\label{eq:region_JD}
\end{align}%

\subsection{Performance Metrics}

In this work, we are interested in performance metrics that are related to the
achievable rate tuple $(R_1,\ldots,R_K)$. For simplicity, we
consider the following utility functionals of the rate tuple. 
\begin{align}
f_m(R_1,\ldots,R_K)=\begin{cases}
\sum\limits_{k\in[K]}R_k, & m=\text{SR},\\
\sum\limits_{k\in[K]}w_kR_k, &m=\text{WSR},\\
\min\limits_{k\in[K]} R_k, &m=\text{WUR}, 
\end{cases}\label{eq:utility_function}
\end{align}
where the coefficient $w_k\ge0$ denotes the weight for user $k$. Here,
$f_{\text{SR}}$, $f_{\text{WSR}}$, and $f_{\text{WUR}}$ represent the
sum rate, the weighted sum rate, and the worst-user rate, respectively. 

We mainly focus on the precoder design problem such that one of the
above rate functions are maximized. Specifically, we shall maximize
these functions respectively over the $2^K-1$ covariance matrices subject to the transmit power
constraint in Section \ref{sec:RS_op_JD}. Another way to the precoder design is to minimize the transmit power for a
given target rate tuple as shown in Section \ref{sec:PM_JD}. We shall show that these problems can be solved using the same optimization method by applying the same transformation on the constraint functions.

It is worth noting that although the sum rate optimization problem is a
special case of the weighted sum rate problem, the optimization
technique and complexity can be very different. That is why we separate
the sum rate problem from the general weighted sum rate problem.

\section{Optimal Precoder Design}\label{sec:op_JD}
In this section, we investigate optimization problems for precoder
design for the general RS scheme under joint decoding. We first consider the rate maximization
problems and then the power minimization problem. 
For convenience, we introduce the following notations on the rates of
the sub-messages $\pmb{R} :=
\left({R}^{(k)}_{\Kc}\right)_{k\in\Kc,\Kc\subseteq[K]}$, the
rates of the re-assembled messages $\tilde{\pmb{R}} :=
\left({R}_{\Kc}\right)_{\Kc\subseteq[K]}$, and the covariance matrices $\pmb{Q} := \left(\Qm_{\Kc}\right)_{\Kc\subseteq[K]}$.

\subsection{Rate Maximization}\label{sec:RS_op_JD}
The rate maximization problems have been widely considered in wireless
communications. For example, such problems have been studied for BC in a
variety of scenarios, i.e., downlink unicast~\cite{unicastBC2008,
unicastBC2011}, downlink multicast~\cite{multicastBC2014}, and multi-group multicast~\cite{multigroupmulticastBC2011}. Consider the following transmit power constraint
\begin{align}
&\sum_{\Kc\subseteq [K]} \trace\left( \Qm_{\Kc} \right) \le P,\label{eq:RSpowerconstraint}
\end{align}
where $P$ is the power budget.
We would like to maximize the utility functions of the rate tuple, $f_m(R_1,\ldots,R_K)$, $m=\text{SR},\text{WSR},\text{WUR}$, subject to the rate constraints in~\eqref{eq:region_JD} and the constraints on the covariance matrices in~\eqref{eq:con_PSD} and~\eqref{eq:RSpowerconstraint}. 
Specifically, we formulate the following general rate maximization  problem
	\begin{align*}
          \Pc_m^{\text{JD}}:\qquad\max_{\pmb{Q},\pmb{R}}
          \quad &f_m(R_1,\ldots,R_K)\\
	\text{s.t.} \quad&\eqref{eq:con_PSD},\;\eqref{eq:region_JD},\;\eqref{eq:RSpowerconstraint},
	\end{align*}
        where $f_m(\cdot)$ is given by \eqref{eq:utility_function},
        $R_k$ in $f_m(\cdot)$ is given by $R_k=\sum_{\Kc\in\collectionK^{(k)}} {R}^{(k)}_{\mathcal{K}}$
        and $R_{\mathcal{K}}$ in \eqref{eq:region_JD} is given by $R_{\mathcal{K}}=\sum_{k\in\mathcal{K}} {R}^{(k)}_{\mathcal{K}}$.
$\Pc_m^{\text{JD}}$ is a nonconvex problem with $M^2(2^K-1)+K 2^{K-1}$ variables and $K
\left(2^{2^{K-1}}-1\right)+2^K$ constraints.
        

Note that $f_m(R_1,\ldots,R_K)$ is concave, the constraints in~\eqref{eq:con_PSD} are convex, and the constraint in~\eqref{eq:RSpowerconstraint} is linear. In addition, \eqref{eq:region_JD} can be rewritten as
\begin{align}
&\underbrace{\sum_{\mathcal{K}\in \collectionS^{(k)}} {R}_{\mathcal{K}} -
\log\biggl(1 +\!\!\!\!\!\sum_{\mathcal{K}'\notin
\collectionK^{(k)}}\!\!\hv_k^\H\Qm_{\mathcal{K}'}\hv_k+\!\!\sum_{\mathcal{K}\in
\collectionS^{(k)}} \!\!\hv_k^\H\Qm_{\mathcal{K}}\hv_k
\biggr)}_{\text{convex}}+\underbrace{\log\biggl(1 + \!\!\!\!\sum_{\mathcal{K}'\notin
\collectionK^{(k)}} \!\!\hv_k^\H\Qm_{\mathcal{K}'}\hv_k\biggr)
}_{\text{concave}} \le 0,\nonumber\\
&\hspace{10cm}\forall\, k\in[K], \,\collectionS^{(k)} \subseteq
\collectionK^{(k)}.\label{eq:major_mini}
\end{align}
Note that each constraint function in~\eqref{eq:major_mini} can be regarded as a difference of two convex functions. Therefore, $\Pc_m^{\text{JD}}$ is a difference of convex functions (DC) programming. A stationary point of $\Pc_m^{\text{JD}}$ can be obtained by  CCCP~\cite{sun2017MM}. The main idea is to solve a sequence of successively refined approximate convex problems, each of which is obtained by linearizing the concave part in~\eqref{eq:major_mini} and preserving the
remaining convexity of $\Pc_m^{\text{JD}}$.
Specifically, at the $i$-th iteration, the derivative of the concave term $\log\biggl(1 + \sum_{\mathcal{K}'\notin
\collectionK^{(k)}} \hv_k^\H\Qm_{\mathcal{K}'}\hv_k\biggr)$ at $\Qm(i-1)$ is given by
\begin{align}
\frac{\partial\log\biggl(1 +\sum_{\mathcal{K}'\notin \collectionK^{(k)}} \!\hv_k^\H\Qm_{\mathcal{K}'}\hv_k\biggr)}{\partial \Qm_{{\Kc}}}\bigg|_{\Qm=\Qm(i-1)}=\frac{\hv_k\hv_k^\H}{\biggl(1 +  \sum_{\mathcal{K}'\notin \collectionK^{(k)}} \!\hv_k^\H\Qm_{\mathcal{K}'}\hv_k\biggr)\ln (2)}, \;\forall\,{{\Kc}}\notin \collectionK^{(k)},\nonumber
\end{align}
where $\pmb{Q}(i-1)$ denotes the optimal solution of the approximate
convex problem at the $(i-1)$-th iteration. Therefore, we can linearize the concave term in \eqref{eq:major_mini} at $\pmb{Q}(i-1)$ as follows
\begin{align}
L_{k}(\Qm;\pmb{Q}(i-1))=\log\biggl(1 +\!\!\sum_{\mathcal{K}'\notin \collectionK^{(k)}} \!\hv_k^\H\Qm_{\mathcal{K}'}(i-1)\hv_k\biggr)+&\frac{\sum_{\mathcal{K}'\notin \collectionK^{(k)}} \hv_k^\H\Big(\Qm_{\mathcal{K}'}-\Qm_{\mathcal{K}'}(i-1)\Big)\hv_k
}{\biggl(1 +\sum_{\mathcal{K}'\notin \collectionK^{(k)}} \hv_k^\H\Qm_{\mathcal{K}'}(i-1)\hv_k\biggr)\ln 2},\nonumber\\
&\forall \,k\in[K].\label{eq:linear1}
\end{align}
 In the following, we shall provide the details of the CCCP for obtaining a stationary point of $\Pc_m^{\text{JD}}$ for $m=\text{SR},\text{WSR},\text{WUR}$, respectively.

First, consider $m=\text{SR}$. Since
\begin{align}
 f_{\text{SR}}(R_1,\ldots,R_K)=\sum_{k\in[K]} R_k = \sum_{k\in[K]}
  \sum_{\Kc\in\collectionK^{(k)}}
  {R}^{(k)}_{\Kc}
  =\sum_{\Kc\subseteq[K]}\sum_{k\in\Kc}
{R}^{(k)}_{\Kc}
 = \sum_{\Kc\subseteq[K]}{R}_{\mathcal{K}},
\end{align}
$\Pc_{\text{SR}}^{\text{JD}}$ can be simplified to
\begin{align}
  \Pc_{\text{SR}}^{\text{JD}}: \qquad\max_{\pmb{Q},\tilde{\pmb{R}}}\quad &\sum_{\Kc\subseteq[K]}{R}_{\Kc}\nonumber\\
  \text{s.t.}\quad&\eqref{eq:con_PSD},\;\eqref{eq:region_JD},\; \eqref{eq:RSpowerconstraint}.\nonumber
\end{align}
Note that the rates of the sub-messages $\pmb{R}$ do not appear in the simplified form of $\Pc_{\text{SR}}^{\text{JD}}$, and the number of variables is reduced from $M^2(2^K-1)+K 2^{K-1}$ to $(M^2+1)(2^K-1)$.
The approximate convex problem of $\Pc_{\text{SR}}^{\text{JD}}$ at the $i$-th iteration is given by
	\begin{align}
	\tilde\Pc_{\text{SR}}^{\text{JD}}(i): \quad\mathop{\max
        }\limits_{\pmb{Q},\tilde{\pmb{R}}}\quad &\sum_{\Kc\subseteq[K]}{R}_{\Kc}\nonumber\\
	\qquad\qquad\;\;\text{s.t.}\quad &\eqref{eq:con_PSD},\;\eqref{eq:RSpowerconstraint}, \nonumber\\
	\qquad\qquad&\sum_{\mathcal{K}\in \collectionS^{(k)}} {R}_{\mathcal{K}}-
	\log\biggl(1 +\!\!\!\!\! \sum_{\mathcal{K}'\notin \collectionK^{(k)}}\!\!\hv_k^\H\Qm_{\mathcal{K}'}\hv_k+\!\!\sum_{\mathcal{K}\in \collectionS^{(k)}} \!\!\hv_k^\H\Qm_{\mathcal{K}}\hv_k  \biggr)\nonumber\\
	&\qquad+L_{k}(\Qm;\pmb{Q}(i-1))\le 0,\quad \forall\,k\in[K],\, \collectionS^{(k)}\subseteq\collectionK^{(k)}.\label{eq:rate_linear}
	\end{align}


Next, consider $m=\text{WSR}$.
$\Pc_{\text{WSR}}^{\text{JD}}$ can be expressed as
	\begin{align}
	\Pc_{\text{WSR}}^{\text{JD}}:\quad\max_{\pmb{Q},\pmb{R}} &\sum_{k\in[K]}w_k\sum_{\Kc\in\collectionK^{(k)}}{R}^{(k)}_{\Kc}\nonumber\\
	\qquad\quad\;\; \text{s.t.} \;&\eqref{eq:con_PSD},\;\eqref{eq:RSpowerconstraint}, \nonumber\\
	&\sum_{\mathcal{K}\in \collectionS^{(k)}} \sum_{k\in\mathcal{K}} {R}^{(k)}_{\mathcal{K}}-
	\log\biggl(1 +\!\!\!\sum_{\mathcal{K}'\notin \collectionK^{(k)}}\!\!\!\hv_k^\H\Qm_{\mathcal{K}'}\hv_k+\!\!\sum_{\mathcal{K}\in \collectionS^{(k)}} \!\!\hv_k^\H\Qm_{\mathcal{K}}\hv_k  \biggr)\nonumber\\
	&+\log\biggl(1 + \!\!\!\!\sum_{\mathcal{K}'\notin \collectionK^{(k)}} \!\!\!\hv_k^\H\Qm_{\mathcal{K}'}\hv_k\biggr)\le 0, \quad\forall\, k\in[K],\, \collectionS^{(k)}\subseteq\collectionK^{(k)}.\label{eq:rate_transform} 
	\end{align}
The approximate convex problem of $\Pc_{\text{WSR}}^{\text{JD}}$ at the $i$-th iteration is given by
	\begin{align}
	\tilde\Pc_{\text{WSR}}^{\text{JD}}(i):\quad\max_{\pmb{Q},\pmb{R}} &\sum_{k\in[K]}w_k\sum_{\Kc\in\collectionK^{(k)}}{R}^{(k)}_{\Kc}\nonumber\\
	\qquad\qquad\quad\;\, \text{s.t.} \; &\eqref{eq:con_PSD},\; \eqref{eq:RSpowerconstraint},\nonumber\\
&\sum_{\mathcal{K}\in \collectionS^{(k)}} \sum_{k\in\mathcal{K}} {R}^{(k)}_{\mathcal{K}}-
	\log\biggl(1 +\!\!\!\sum_{\mathcal{K}'\notin \collectionK^{(k)}}\!\!\!\hv_k^\H\Qm_{\mathcal{K}'}\hv_k+\!\!\sum_{\mathcal{K}\in \collectionS^{(k)}} \!\!\hv_k^\H\Qm_{\mathcal{K}}\hv_k  \biggr)\nonumber\\
&\quad\qquad+L_{k}(\Qm;\pmb{Q}(i-1))\le 0,\qquad \forall\,k\in[K],\, \collectionS^{(k)}\subseteq\collectionK^{(k)}. \label{eq:rate_transform_linear}
\end{align}

Then, consider $m=\text{WUR}$. We introduce an extra slack variable $y$ which serves as a lower bound of $\min\limits_{k\in[K]} R_k=\min\limits_{k\in[K]} \sum_{\Kc\in\collectionK^{(k)}}{R}^{(k)}_{\Kc}$, i.e.,
\begin{align}
y \le \sum_{\Kc\in\collectionK^{(k)}}{R}^{(k)}_{\Kc},\quad \forall \, k\in [K]. \label{eq:y_cons}
\end{align}
Thus, $\Pc_{\text{WUR}}^{\text{JD}}$ can be equivalently transformed to the following problem
	\begin{align*}
	\Pc_{\text{WUR}}^{\text{JD}}:\qquad\quad\max_{\pmb{Q},\pmb{R},y}\ &y \\
	\text{s.t.} \quad& \eqref{eq:con_PSD},\;\eqref{eq:RSpowerconstraint},\;\eqref{eq:rate_transform},\;\eqref{eq:y_cons}. 
	\end{align*}
The number of variables in $\Pc_{\text{WUR}}^{\text{JD}}$ becomes $M^2(2^K-1)+K 2^{K-1}+1$.
The approximate convex problem of $\Pc_{\text{WSR}}^{\text{JD}}$ at the $i$-th iteration is given by
	\begin{align*}
	\tilde\Pc_{\text{WUR}}^{\text{JD}}(i):\qquad\quad\max_{\pmb{Q},\pmb{R},y}\ &y\\
	\text{s.t.} \quad& \eqref{eq:con_PSD},\;\eqref{eq:RSpowerconstraint},\;\eqref{eq:rate_transform_linear},\;\eqref{eq:y_cons}. 
	\end{align*}

Finally, the details of the CCCP for obtaining a stationary point of
$\Pc_{m}^{\text{JD}}$, for $m=\text{SR},\text{WSR},\text{WUR}$,
respectively, are summarized in Algorithm~\ref{algo:max_sr} in which 
$J_m(i)$ is given by
\begin{align}
J_m(i)\triangleq\begin{cases}
\left\Vert(\pmb{Q}(i),\tilde{\pmb{R}}(i))-(\pmb{Q}(i-1),\tilde{\pmb{R}}(i-1))\right\Vert_2, & m=\text{SR},\\
\left\Vert \left(\pmb{Q}(i),\pmb{R}(i)\right)-\left(\pmb{Q}(i-1),\pmb{R}(i-1)\right)\right\Vert_2, &m=\text{WSR},\\
\left\Vert \left(\pmb{Q}(i),\pmb{R}(i),y(i)\right)-\left(\pmb{Q}(i-1),\pmb{R}(i-1),y(i-1)\right)\right\Vert_2, &m=\text{WUR}.
\end{cases}\label{eq:convergence criterion}
\end{align}
\begin{algorithm}[h]
	\caption{Obtaining A Stationary Point of $\Pc_{m}^{\text{JD}}$}
	\label{algo:max_sr}
	\begin{algorithmic}[1]
		\STATE {Choose any feasible covariance matrices $\pmb{Q}(0)$ of $\Pc_{m}^{\text{JD}}$, and set $i=1$.} 
		\REPEAT 
		\STATE Obtain an optimal solution of $\tilde{\Pc}_{m}^{\text{JD}}(i)$, denoted by $\{(\pmb{Q}(i),\tilde{\pmb{R}}(i))\}$, $\{(\pmb{Q}(i),\pmb{R}(i))\}$ and $\{(\pmb{Q}(i),\pmb{R}(i),y(i))\}$ for $m=\text{SR},\text{WSR},\text{WUR}$, respectively, with an interior point method.
		\STATE Set $i=i+1$.
		\UNTIL{{ the convergence criterion $J_m(i)\leq\epsilon$ is met.}} 
	\end{algorithmic}
\end{algorithm}
\begin{claim}[Convergence of Algorithm \ref{algo:max_sr}]\label{lemma:converge}
	As $i\to\infty$, $\{(\pmb{Q}(i),\tilde{\pmb{R}}(i))\}$, $\{(\pmb{Q}(i),\pmb{R}(i))\}$ as well as $\{(\pmb{Q}(i),\pmb{R}(i),y(i))\}$ obtained by Algorithm \ref{algo:max_sr} converge to a stationary point\footnote{ 
	A stationary point is a point that satisfies the necessary optimality conditions (for example the KKT conditions) of a nonconvex optimization problem~\cite{sun2017MM}, and is the classic goal of designing iterative algorithms for solving a nonconvex optimization problem, as there is no effective method for solving a general nonconvex problem optimally.} of $\Pc_{m}^{\text{JD}}$ for $m=\text{SR},\text{WSR},\text{WUR}$, respectively.
\end{claim}
\begin{proof}
	We have shown that $\Pc_{m}^{\text{JD}}$ is a DC programming and we propose to solve it with CCCP. It has been validated in \cite{sun2017MM} that solving DC programming through CCCP always returns a stationary point.  
\end{proof}
Note that Algorithm \ref{algo:max_sr}, based on CCCP, usually converges faster than conventional gradient methods, as it exploits the partial concavity of $\Pc_{m}^{\text{JD}}$. 
By~\cite{facchinei2017ghost}, we know that the number of iterations of Algorithm \ref{algo:max_sr} does not scale with the problem size. Thus, the computational complexity order for Algorithm \ref{algo:max_sr} is the same as that for solving $\tilde\Pc_{m}^{\text{JD}}(i)$ in Step 3.
When an interior point method is applied, the computational complexity for solving $\tilde\Pc_{\text{SR}}^{\text{JD}}(i)$ is $\mathcal O\left(M^4K^{1.5}2^{0.75\times2^{K}+2K}\right)$, and the computational complexities for solving $\tilde\Pc_{\text{WSR}}^{\text{JD}}(i)$ and $\tilde\Pc_{\text{WUR}}^{\text{JD}}(i)$ are $\mathcal O\left((M^2+K)^2K^{1.5} 2^{0.75\times2^{K}+2K}\right)$~\cite{boyd2004convex}.


The initial value for $\pmb{Q}(0)$ can be chosen randomly (provided that feasibility is ensured) or through a heuristic method.
One of the possible initial values for the covariance matrices can be the ZF precoder or the MMSE precoder used in \cite{li2018linearly}. 
	In practice, we can run Algorithm \ref{algo:max_sr} multiple times with different feasible initial points $\pmb{Q}(0)$ to obtain multiple stationary points, and choose the stationary point with the best objective value as a suboptimal solution.



\subsection{Power Minimization}\label{sec:PM_JD}

Another relevant problem in wireless communications is power efficiency optimization,
i.e., to minimize the transmit power for a given target rate
tuple~$(r_1,\ldots,r_K)$. 
Such power minimization problem has been studied extensively for BC in a
variety of communication scenarios, i.e., downlink unicast
\cite{bengtsson2001optimal, schubert2004solution}, downlink multicast
\cite{sidiropoulos2006transmit} and multi-group multicast
\cite{karipidis2008quality}. Furthermore, precoder design that minimizes
the total power consumption while ensuring target user rates is also
studied in emerging scenarios such as large-scale multi-cell multi-user
MIMO systems \cite{sanguinetti2016large}, or in the presence of
eavesdroppers \cite{liao2011qos}. 
Considering the following rate constraints
	\begin{align}
	&\sum_{\Kc\in\collectionK^{(k)}} {R}_{\mathcal{K}}^{(k)}\ge r_k,\quad \forall\,k\in[K], \label{eq:pri_target}
	\end{align}
where $r_k$, $k\in[K]$, are the target rates. We would like to minimize the transmit power subject to the rate constraints in~\eqref{eq:region_JD} and~\eqref{eq:pri_target} and the constraints on the covariance matrices in \eqref{eq:con_PSD}. Specifically, we formulate the
power minimization problem for the general RS scheme as follows
	\begin{align*}
	\Pc_{\text{PM}}^{\text{JD}}:\quad\;\mathop {\min }\limits_{\pmb{Q},\pmb{R}} &\sum_{\Kc\subseteq[K]}\trace(\Qm_{\Kc})\\
	\text{s.t.} \quad &\;\;\eqref{eq:con_PSD},\;\eqref{eq:region_JD},\;  \eqref{eq:pri_target}.
	\end{align*}
	
As in the rate maximization problems presented previously, the above
power minimization problem can also be regarded as a DC programming. 
Hence, a stationary point of $\Pc_{\text{PM}}^{\text{JD}}$ can be obtained using CCCP. Specifically, at the $i$-th iteration, the approximate convex problem of $\Pc_{\text{PM}}^{\text{JD}}$ is given by
\begin{align*}
	\tilde{\Pc}_{\text{PM}}^{\text{JD}}(i):&\qquad\mathop {\min
        }\limits_{\pmb{Q},\pmb{R}} \quad \sum_{\Kc\subseteq[K]}\trace(\Qm_{\Kc})\\
	&\qquad\;\;\text{s.t.} \quad \eqref{eq:con_PSD},\;\eqref{eq:rate_transform_linear},\;\eqref{eq:pri_target}.
\end{align*}
The complete algorithm and its convergence proof are similar to those of $\Pc_{m}^{\text{JD}}$. Thus, we omit the details due to space limitation.


\section{Practical Considerations for Future Implementation}\label{sec:JD_reduced}
\newcommand{\collectionU}{\boldsymbol{\mathcal{U}}}
In the previous section, we have formulated the precoder design problems
for the general RS scheme with joint decoding, and proposed iterative algorithms for
their solutions. Nevertheless, there are a number of challenges for its
implementation for possible applications in future wireless networks.  

First, the number of streams $2^K-1$ can be large when $K$ is large,
which increases significantly the precoding and decoding complexity as
compared to the case with only $K$ private streams. Second, the number
of rate constraints for joint decoding can be as large as $K
\left(2^{2^{K-1}}-1\right)$ as one can verify from
\eqref{eq:region_JD}. The computational complexity of the proposed precoder design
algorithms can be formidable with a large $K$. Third, perfect CSIT may be
hard to obtain in practice. We should account for the CSIT inaccuracy in
our design. 

To address the above challenges, we propose the following solutions: 
\begin{itemize}
	\item Use successive decoding to reduce the decoding complexity at the receivers' side.
	\item Apply a stream selection algorithm to reduce the computational complexity at the BS and to further reduce the decoding complexity at the receivers' side.
	\item Adjust the current precoder design algorithms so that the CSIT
	error is taken into account. 
\end{itemize}

\subsection{Successive Decoding}
\label{sec:op_SIC}
In this subsection, we consider the precoder design problems with successive decoding. For $k\in[K]$, let us assume that the $2^{K-1}$ elements in
$\collectionK^{(k)}$ is somehow ordered such that the $n$-th element is
$\Kc^{(k)}_n$, $n\in[2^{K-1}]$. 
We define $\pmb{\pi}^{(k)} := \bigl(\pi_{k,1}, \ldots,
\pi_{k,2^{K-1}}\bigr)$ as the permutation vector of length
$2^{K-1}$, which is used to specify the decoding order at
user~$k$. 
Specifically, at round $n\in[2^{K-1}]$, each user $k$ decodes
stream~$\Kc^{(k)}_{\pi_{k,n}}$ by treating
streams~$\{\Kc^{(k)}_{\pi_{k,n'}}:\,n'>n\}$ as noise. For a given
decoding order $\pmb{\pi}^{(k)}$, the achievable
rate region is thus defined by the following constraints
\begin{align}
{R}_{\Kc^{(k)}_{\pi_{k,n}}} \le
\log\biggl(1+\frac{\hv_k^H\Qm_{\Kc^{(k)}_{\pi_{k,n}}}\hv_k}{1+\sum_{{\Kc'}\not\in
		\collectionK^{(k)}}\hv_k^H\Qm_{{\Kc'}}\hv_k+\sum_{n'>n}\hv_k^H\Qm_{\Kc^{(k)}_{\pi_{k,n'}}}\hv_k}\biggr),
\quad \forall\,k\in[K], n\in[2^{K-1}].\label{region_SIC}
\end{align}

Let us introduce the notation on the decoding orders $\pmb{\pi} := \left(\pmb{\pi}^{(k)}\right)_{k\in[K]}$. The rate maximization problem can be formulated as follows
\begin{align}
\Pc_m^{\text{SD}}:\quad\mathop {\max
}\limits_{\pmb{\Qm},\pmb{R}, \pmb{\pi}}\;&f_m(R_1,\ldots,R_K)\nonumber\\
\quad\text{s.t.}\quad  & \eqref{eq:con_PSD},\;\eqref{eq:RSpowerconstraint},\;\eqref{region_SIC}.\nonumber
\end{align}
$\Pc_m^{\text{SD}}$ is a challenging mixed discrete-continuous
optimization problem with $(M^2+1)(2^K-1)$ continuous variables
$\pmb{\Qm}$ and $\pmb{R}$, $(2^{K-1}!)^K$ possible values for the
discrete variable $\pmb{\pi}$ and $K\left(2^{K-1}\right)+2^K$
constraints. One straightforward way to solve $\Pc_m^{\text{SD}}$ is to
first solve the optimization with respect to $\pmb{\Qm}$ and $\pmb{R}$
for a given $\pmb{\pi}$, denoted by $\Pc_{m}^{\text{SD}}(\pmb{\pi})$, and then solve the optimization with respect to $\pmb{\pi}$ using exhaustive search over all $(2^{K-1}!)^K$ possible values for $\pmb{\pi}$. 


First, we solve $\Pc_{{m}}^{\text{SD}}(\pmb{\pi})$ for a given $\pmb{\pi}$ using CCCP.
To avoid redundancy, we present in detail the following sum rate maximization under successive decoding, optimization under the other two rate criteria can be followed analogously. The transmit power minimization can also be handled similarly. For a given
decoding order ${\pmb{\pi}}$, the sum rate optimization under successive decoding is formulated as
\begin{align}
\Pc_{\text{SR}}^{\text{SD}}(\pmb{\pi}): \qquad\mathop {\max }\limits_{\pmb{\Qm},\tilde{\pmb{R}} } &\sum_{\Kc\subseteq[K]}{R}_{\Kc}\nonumber\\
\qquad\qquad\quad \text{s.t.}\quad& \eqref{eq:con_PSD},\;\eqref{eq:RSpowerconstraint},\;\eqref{region_SIC}.\nonumber
\end{align}
$\Pc_{\text{SR}}^{\text{SD}}(\pmb{\pi})$ is a nonconvex problem with $(M^2+1)(2^K-1)$ variables and $K
\left(2^{K-1}\right)+2^K$ constraints.
Note that the objective function and the constraint in~\eqref{eq:RSpowerconstraint} are linear, and the constraints in~\eqref{eq:con_PSD} are convex. In addition, \eqref{region_SIC} can be rewritten as
\begin{align}
&\underbrace{{R}_{\Kc^{(k)}_{\pi_{k,n}}}-	\log\biggl(1+\hv_k^H\Qm_{\Kc^{(k)}_{\pi_{k,n}}}\hv_k+\sum_{{\Kc'}\not\in		\collectionK^{(k)}}\hv_k^H\Qm_{{\Kc'}}\hv_k+\sum_{n'>n}\hv_k^H\Qm_{\Kc^{(k)}_{\pi_{k,n'}}}\hv_k\biggr)}_{\text{convex}}\nonumber\\
&+\underbrace{\log\biggl(1 +\!\!\!\!\sum_{{\Kc'}\not\in \collectionK^{(k)}} \!\!\!\!\hv_k^\H\Qm_{\mathcal{K}'}\hv_k+\sum_{n'>n}\hv_k^H\Qm_{\Kc^{(k)}_{\pi_{k,n'}}}\hv_k\biggr)}_{\text{concave}} \le 0, \quad\forall\, k\in[K], \,n\in[2^{K-1}],\label{eq:major_mini_SIC}
\end{align}
and be regarded as a difference of two convex functions. Similarly to $\Pc_{\text{SR}}^{\text{JD}}$ in Section \ref{sec:RS_op_JD}, $\Pc_{\text{SR}}^{\text{SD}}(\pmb{\pi})$ is a DC programming and a stationary point of $\Pc_{\text{SR}}^{\text{SD}}(\pmb{\pi})$ can be obtained using CCCP. The approximate convex problem of $\Pc_{\text{SR}}^{\text{SD}}(\pmb{\pi})$ at the $i$-th iteration is given by
\begin{align}
\tilde{\Pc}_{\text{SR}}^{\text{SD}}(\pmb{\pi},i): \quad\mathop {\max }\limits_{\pmb{\Qm},\pmb{\tilde{R}}} &\sum_{\Kc\subseteq[K]}R_{\Kc}\nonumber\\
\qquad\qquad\quad \text{s.t.} \quad&\eqref{eq:con_PSD},\;\eqref{eq:RSpowerconstraint},\nonumber\\
&{R}_{\Kc^{(k)}_{\pi_{k,n}}}-
	\log\biggl(1+\hv_k^H\Qm_{\Kc^{(k)}_{\pi_{k,n}}}\hv_k+\sum_{{\Kc'}\not\in
		\collectionK^{(k)}}\hv_k^H\Qm_{{\Kc'}}\hv_k+\sum_{n'>n}\hv_k^H\Qm_{\Kc^{(k)}_{\pi_{k,n'}}}\hv_k\biggr)\nonumber\\
&\qquad\qquad+L_{k,n}^{\text{SD}}(\Qm;\pmb{Q}(i-1))\le 0,\quad  \forall\,k\in[K], n\in[2^{K-1}],\label{eq:rate_linear_SIC}
\end{align}
\noindent where $L_{k,n}^{\text{SD}}(\Qm;\pmb{Q}(i-1))$ corresponds to the linearization of the concave term in \eqref{eq:major_mini_SIC} at $\pmb{Q}(i-1)$, and is given by
\begin{align}
&L_{k,n}^{\text{SD}}(\Qm;\pmb{Q}(i-1))=\log\biggl(1 +\!\!\!\!\sum_{{\Kc'}\not\in
	\collectionK^{(k)}} \!\!\!\!\hv_k^\H\Qm_{\mathcal{K}'}(i-1)\hv_k+\sum_{n'>n}\hv_k^H\Qm_{\Kc^{(k)}_{\pi_{k,n'}}}(i-1)\hv_k\biggr)\nonumber\\
&\qquad+\frac{\sum_{{\Kc'}\not\in
		\collectionK^{(k)}} \hv_k^\H\biggl(\Qm_{\mathcal{K}'}-\Qm_{\mathcal{K}'}(i-1)\biggr)\hv_k+\sum_{n'>n}\hv_k^H\biggl(\Qm_{\Kc^{(k)}_{\pi_{k,n'}}}-\Qm_{\Kc^{(k)}_{\pi_{k,n'}}}(i-1)\biggr)\hv_k
}{\biggl(1 +\sum_{{\Kc'}\not\in
	\collectionK^{(k)}} \hv_k^\H\Qm_{\mathcal{K}'}(i-1)\hv_k+\sum_{n'>n}\hv_k^H\Qm_{\Kc^{(k)}_{\pi_{k,n'}}}(i-1)\hv_k\biggr)\ln (2)},\nonumber\\
&\qquad\qquad\qquad\forall\,k\in[K], n\in[2^{K-1}],\label{eq:linear_SIC}
\end{align}
with $\{\Qm(i-1)\}$ being the optimal solution of $\tilde\Pc_{\text{SR}}^{\text{SD}}(\pmb{\pi},i-1)$ at the $(i-1)$-th iteration.

Then, we can perform exhaustive search over all $(2^{K-1}!)^K$ possible decoding orders. The details are summarized in Algorithm \ref{algo:max_sr_SIC}.\footnote{Note that there is no guarantee that $\pi^\dagger$ is the optimal decoding order, as we can obtain only a stationary point of $\Pc_{\text{SR}}^{\text{SD}}(\pmb{\pi})$.}
\begin{algorithm}
		\caption{Solving $\Pc_{\text{SR}}^{\text{SD}}$ with Exhaustive Search}
		\label{algo:max_sr_SIC}
		\begin{algorithmic}[1]
			\STATE Set $R_{\text{SR}}^{\dagger}$=0.
			\FOR{all possible values for $\pmb{\pi}$}
			\STATE {Choose any feasible covariance matrices $\pmb{Q}(0)$ of $\Pc_{\text{SR}}^{\text{SD}}(\pmb{\pi})$, and set $i=1$.} 
			\REPEAT 
			\STATE Obtain an optimal solution of $\tilde{\Pc}_{\text{SR}}^{\text{SD}}(\pmb{\pi},i)$, denoted by $\{(\pmb{Q}(i),\tilde{\pmb{R}}(i))\}$, with an interior point method.
			\STATE Set $i=i+1$.
			\UNTIL{{the convergence criterion $J_{\text{SR}}(i)\leq\epsilon$ is met.}}
			\IF{$\sum_{\Kc\subseteq[K]}{R}_{\Kc}(i-1)>R_{\text{SR}}^{\dagger}$}
			\STATE Set $R_{\text{SR}}^{\dagger}=\sum_{\Kc\subseteq[K]}{R}_{\Kc}(i-1)$, $\pmb{Q}^{\dagger}=\pmb{Q}(i-1)$, $\tilde{\pmb{R}}^{\dagger}=\tilde{\pmb{R}}(i-1)$ and $\pmb{\pi}^{\dagger}=\pmb{\pi}$.
			\ENDIF
			\ENDFOR
		\end{algorithmic}
\end{algorithm}
\vspace{-0.5cm}
\begin{claim}[Convergence of solving $\Pc_{\text{SR}}^{\text{SD}}(\pmb{\pi})$ with CCCP]
For any $\pmb{\pi}$, $\{(\pmb{Q}(i),\tilde{\pmb{R}}(i))\}$ obtained by Step 3 to Step 6 in Algorithm \ref{algo:max_sr_SIC} converges to a stationary point of $\Pc_{\text{SR}}^{\text{SD}}(\pmb{\pi})$, as $i\to\infty$.
\end{claim}	
Similarly, the computational complexity
for solving $\tilde\Pc_{\text{SR}}^{\text{SD}}(\pmb{\pi})$ with a
given $\pmb{\pi}$ using CCCP is $\mathcal
O\left(M^6K^{0.5}2^{3.5K}\right)$~\cite{boyd2004convex}. As the number
of all possible choices for $\pmb{\pi}$ is $(2^{K-1}!)^K$, the
computational complexity for Algorithm \ref{algo:max_sr_SIC} is
$\mathcal O\left(M^6K^{0.5}2^{3.5K}(2^{K-1}!)^K\right)$. Although
performing Algorithm \ref{algo:max_sr_SIC} brings higher computational complexity than
performing Algorithm \ref{algo:max_sr} at the BS, successive decoding
has much lower implementation cost at the receivers' side. Since the BS
has much higher computing capability than the receivers, trading extra complexity at the BS for reduced decoding complexity at the receivers seems to be a right choice.

\subsection{Stream Selection}\label{sec:stream_sel}

The prohibitively high computational complexity for Algorithm \ref{algo:max_sr_SIC} at large $K$ comes from the excessive number of all possible streams and the exhaustive search of decoding order.
In fact, we found out by numerical experiments that in most of
the time not all the streams are needed to achieve the best rate
performance. In Fig.~\ref{fig:beam_statistics}, we plot a histogram of
the number of active streams for the precoder design obtained by  Algorithm~\ref{algo:max_sr}. 
We see that it is rare that we need to activate all
the streams.  Indeed, as illustrated in Fig.~\ref{fig:beam_example}, if
a user~(user~3) is sufficiently well separated from the others~(users 1
and 2) in the space domain, no common message should be shared between
this user and the others. In this example, $\rv{M}_{\{1,3\}}$, $\rv{M}_{\{2,3\}}$, and
$\rv{M}_{\{1,2,3\}}$ should be eliminated. Further, in
\cite{li2018linearly} the authors have demonstrated that even when the
optimal solution activates all the streams, removing some streams only
incurs a marginal rate loss. The complexity reduction brought by such
operation is significant and thus appealing for practical
implementation.  
\begin{figure*}
  \begin{minipage}{0.6\textwidth}
	\centering
	\includegraphics[width=0.95\textwidth]{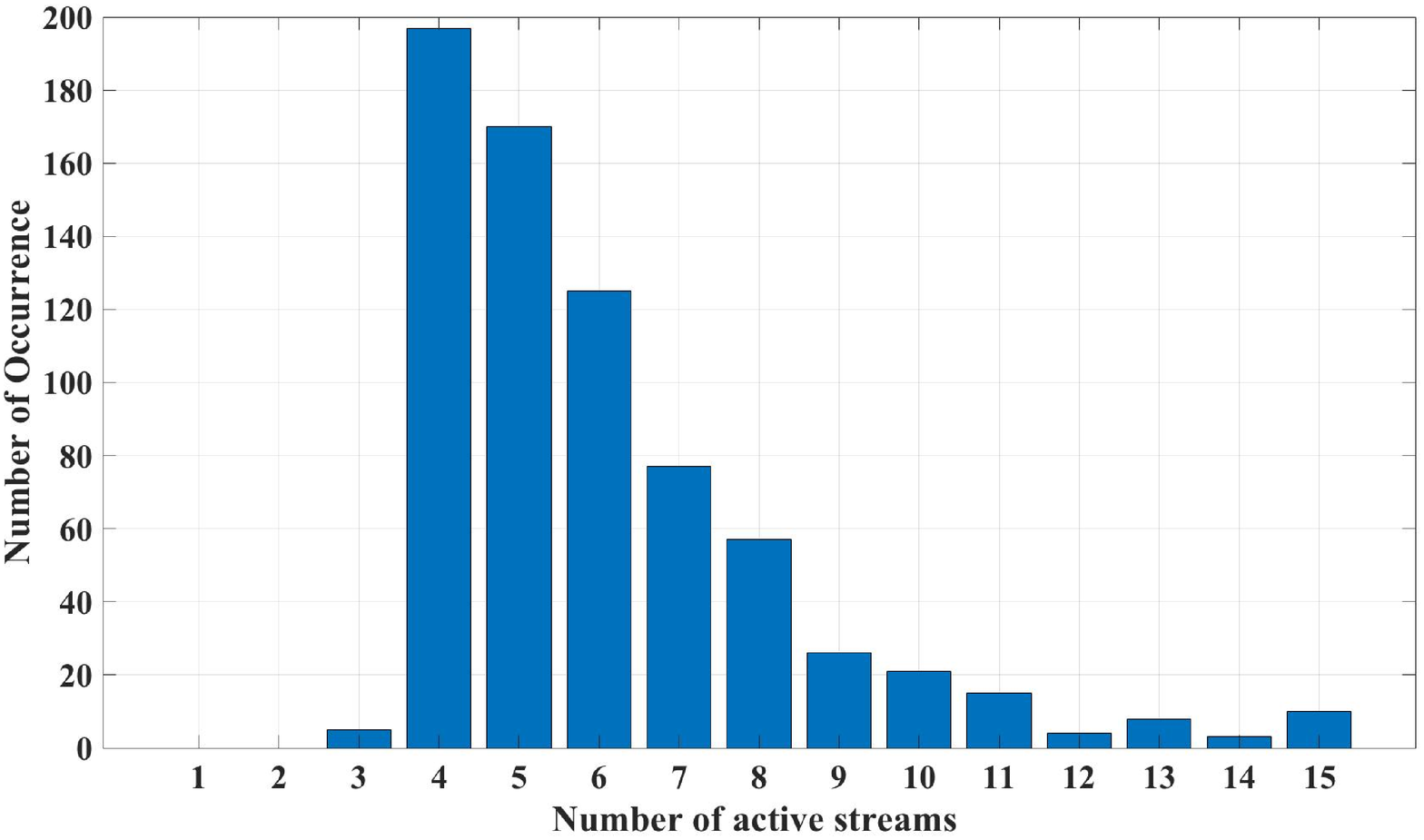}
	\caption{Histogram of the number of active streams to achieve maximum sum rate in the general RS scheme under i.i.d. Rayleigh fading channels, $K=M=4,P=20$ dB.}\label{fig:beam_statistics}
\end{minipage}
\hfill
  \begin{minipage}{0.35\textwidth}
    \vfill
	\centering
	\includegraphics[width=0.8\textwidth]{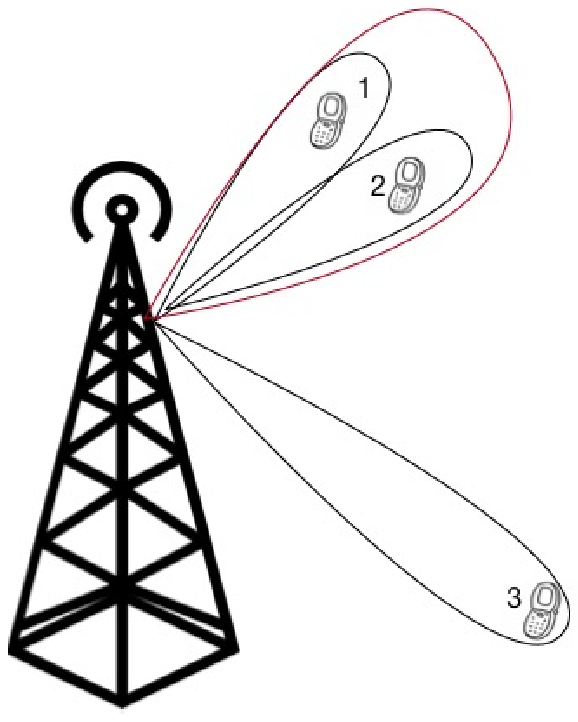}
	\caption{An example to demonstrate the idea that some streams are redundant.}\label{fig:beam_example}
\end{minipage}
\end{figure*}

In this subsection, our goal is to find an efficient way to identify a
``small'' number of ``good'' streams. We propose a stream selection algorithm, which consists of two steps. The first step is to reduce the number of streams from $2^K-1$ to a relatively small number. In the second step, from the remaining streams returned by the first step, we further select \emph{maximum non-overlapping collections} of streams. Maximum non-overlapping collection will be defined in detail later. This stream selection algorithm considerably facilitates the implementation of successive decoding.

In the first step, we apply the stream elimination algorithm (SEA) in \cite{li2018linearly} to reduce the number of streams to smaller than $N_{\text{SEA}}$ in total\footnote{Note that SEA is not limited to successive decoding as presented in this section, it can be also applied under joint decoding to reduce the number of constraints and variables. We omit the details with joint decoding due to space limitation.}. SEA has been proposed to eliminate some of the streams
without losing more than a constant number of bits PCU. Alternatively,
SEA provides a way to gradually eliminate the streams so
as to identify a given number of remaining streams. This is practically
relevant since the number of remaining streams represents the level of
precoding/decoding complexity. Let the collection $\collectionS$ be the set of the remaining streams returned by SEA. 
\begin{claim}
	If the set $\collectionS$ guarantees that the original rate
	region~(without stream elimination) is achievable to within $N$ bits
	PCU, then the original sum rate, weighted sum rate, and worst-user rate
	are achievable to within $KN$, $\sum_k w_k N$, and $N$ bits PCU,
	respectively. 
\end{claim}
\begin{proof}
	From the assumption, if a rate tuple $(R_1,\ldots,R_K)$ is achievable without stream
	elimination, then $(R_1-N, \ldots, R_K-N)$ is achievable with the
	streams in $\collectionS$. Then, the conclusion is straightforward.   
\end{proof}
The above claim implies that if SEA has some
performance guarantee in terms of rate gap, then a similar rate
guarantee can be obtained in terms of the rate functions of our
interest. Since the rate gap is an upper bound, it may appear large for
practical uses. Nevertheless, our numerical results show that such
guarantee provides meaningful improvement even in practical scenarios.  

In the second step, we further select maximum non-overlapping collections of $\collectionS$. Without loss of generality,
$\collectionS$ can be partitioned according to the cardinality of the
sets inside it, namely, 
\begin{equation}
\collectionS = \bigcup_{k=1}^K \collectionS_k,\label{col_partition}
\end{equation}%
where $\collectionS_k$ only contains sets of cardinality $k$. We also refer to $\collectionS_k$  as layer $k$ streams in the following.  
Let $\tilde{\collectionS} \subseteq \collectionS$ be a sub-collection,
and let us also partition $\tilde{\collectionS}$
according to the cardinality of its sets, i.e., $\tilde{\collectionS} = \bigcup_{k=1}^K
\tilde{\collectionS}_k$. We are interested in $\tilde{\collectionS}$
such that $\tilde{\collectionS}_k$ is a \emph{maximum non-overlapping} collection of the sets of $\collectionS_k$, for all $k\in[K]$.
\begin{definition}\label{def:nonoverlap}
For all $k\in[K]$, $\tilde{\collectionS}_k$ is a \emph{maximum non-overlapping} collection of
the sets of $\collectionS_k$, if the following constraints are satisfied
\begin{itemize}
	\item $\Kc_1\cap\Kc_2 = \emptyset$, for any distinct sets $\Kc_1$
	and $\Kc_2$ in $\tilde{\collectionS}_k$;
	\item if $\hat{\collectionS}_k := \collectionS_k \setminus \tilde{\collectionS}_k \ne
	\emptyset$, then $\Kc_1\cap\Kc_2 \ne \emptyset$ for some
	$\Kc_1\in \tilde{\collectionS}_k$ and 
	$\Kc_2\in \hat{\collectionS}_k$. 
\end{itemize}
\end{definition}
It is possible that there are more than one maximum
non-overlapping collection for each layer. Let $D_k$ denote the number
of maximum non-overlapping collections for $\collectionS_k$, for all $k\in[K]$. Then there are $\prod_k D_k$ different $\tilde{\collectionS}$ in total,  which are assembled in $\collectionU$. Each $\tilde{\collectionS}$ suggests a possible set of active streams for transmission. Algorithm \ref{algo:stream_selection} summarizes the above two-step stream selection procedure. Now let us focus on solving $\Pc_{\text{SR}}^{\text{SD}}$ for a given $\tilde{\collectionS}$. For each $\tilde{\collectionS}\in\collectionU$, the following claim can be easily shown.
\begin{algorithm}
	\caption{ Stream Selection Algorithm}
		\label{algo:stream_selection}
		\begin{algorithmic}[1]
			\STATE {Set a target number $N_{\text{SEA}}$ and use SEA in
				\cite{li2018linearly} to reduce the $2^K-1$ streams to at most
				$N_{\text{SEA}}$ streams, denoted by the collection $\collectionS$.} 
			\STATE Construct $\collectionS_k, k\in [K],$ by 
			partitioning $\collectionS$ according to \eqref{col_partition}.
			\STATE For all $k\in[K]$, construct the set of $D_k$ possible maximum non-overlapping collections $\tilde{\collectionS}_k$, denoted by ${\collectionU}_k$.
			\STATE Assemble the $\prod_k D_k$ possible $\tilde{\collectionS}$ in the set  ${\collectionU} = \left\{\bigcup_{k=1}^K
			\tilde{\collectionS}_k| \tilde{\collectionS}_k\in{\collectionU}_k, k\in[K]\right\}$. 
	\end{algorithmic}
\end{algorithm}
\vspace{-0.5cm}
\begin{claim}
	For all $k\in[K]$, let $\tilde{\collectionK}^{(k)} := {\collectionK}^{(k)} \cap
	\tilde{\collectionS}$ where ${\collectionK}^{(k)}
	:=\{\Kc\subseteq[K]:\,k\in\Kc\}$ and $\tilde{\collectionS}$ is a maximum
	non-overlapping collection. Then, $|\Kc_1|\ne|\Kc_2|$ for any  
	distinct sets $\Kc_1$ and $\Kc_2$ in $\tilde{\collectionK}^{(k)}$. 
\end{claim}
Note that each user $k\in[K]$ only needs to decode the streams in
$\tilde{\collectionK}^{(k)}$. According to the above claim, there is at most
one stream per layer, which implies that each user can decode all the
streams successively with the following rule. 
\newtheorem*{SDrule*}{Successive decoding rule}\label{claim:SIC_order}
\begin{SDrule*}
	For all $k\in[K]$, let $\Kc_1,\Kc_2,\ldots\in
	\tilde{\collectionK}^{(k)} := {\collectionK}^{(k)} \cap \tilde{\collectionS}$ be
	ordered with descending cardinality. Then user~$k$ decodes the streams $\Kc_1,\Kc_2,\ldots$ in order.
	
\end{SDrule*}
Let us introduce the notation $\pmb{\tilde{\pi}}^{(k)}_{\tilde{\collectionS}}$ to specify the above unique decoding order at user $k$ for a given $\tilde{\collectionS}$. We further define the decoding order $\pmb{\tilde{\pi}}_{\tilde{\collectionS}} := \left(\pmb{\tilde{\pi}}^{(k)}_{\tilde{\collectionS}}\right)_{k\in[K]}$, which is particular for the given $\tilde{\collectionS}$ since each element $\pmb{\tilde{\pi}}^{(k)}_{\tilde{\collectionS}}$ is unique. Therefore, for each possible $\tilde{\collectionS}$, we can solve $\Pc_{\text{SR}}^{\text{SD}}(\pmb{\tilde{\pi}}_{\tilde{\collectionS}})$ with $\tilde{\collectionK}^{(k)}:=\collectionK^{(k)}\cap\tilde{\collectionS}$ replacing $\collectionK^{(k)}$ using CCCP, and return a stationary point. By now, we can summarize the details for solving $\Pc_{\text{SR}}^{\text{SD}}$ with the proposed stream selection in Algorithm \ref{algo:optimize_with_stream_selection}. The following example can help clarify the above stream selection and
successive decoding rule.
\begin{algorithm}
	\caption{Solving $\Pc_{\text{SR}}^{\text{SD}}$ with The Proposed Stream Selection}
		\label{algo:optimize_with_stream_selection}
		\begin{algorithmic}[1]
			\STATE Generate $\collectionU$ according to Algorithm \ref{algo:stream_selection}.
			\STATE Set $R_{\text{SR}}^{\diamond}$=0.
			\FOR{each $\tilde{\collectionS}\in\collectionU$}
				\REPEAT 
			\STATE Obtain an optimal solution of $\Pc_{\text{SR}}^{\text{SD}}(\pmb{\tilde{\pi}}_{\tilde{\collectionS}},i)$ with $\tilde{\collectionK}^{(k)}:=\collectionK^{(k)}\cap\tilde{\collectionS}$ replacing $\collectionK^{(k)}$, denoted by $\{(\pmb{Q}(i),\tilde{\pmb{R}}(i))\}$, with an interior point method.
			\STATE Set $i=i+1$.
			\UNTIL{{the convergence criterion $J_{\text{SR}}(i)\leq\epsilon$ is met.}}
			\IF{$\sum_{\Kc\subseteq[K]}{R}_{\Kc}(i-1)>R_{\text{SR}}^{\diamond}$}
			\STATE Set $R_{\text{SR}}^{\diamond}=\sum_{\Kc\subseteq[K]}{R}_{\Kc}(i-1)$, $\pmb{Q}^{\diamond}=\pmb{Q}(i-1)$, $\tilde{\pmb{R}}^{\diamond}=\tilde{\pmb{R}}(i-1)$, $\tilde{\collectionS}^{\diamond}=\tilde{\collectionS}$ and  $\pmb{\pi}^{\diamond}=\pmb{\tilde{\pi}}_{\tilde{\collectionS}}$.
			\ENDIF
			\ENDFOR
	\end{algorithmic}
\end{algorithm}

\begin{example}
	Let us consider an example with $K=5$ users. 
	Assume that the initial collection ${\collectionS}$ returned by SEA is 
	\begin{equation*}
	\collectionS = \bigl\{ \{1\},\{2\},\{3\},\{4\},\{5\},\{1,2\}, \{2,3\}, \{3,4\},\{4,5\} ,\{1,2,3\},\{3,4,5\}\bigr\}.
	\end{equation*}
	Therefore, $\collectionS$ can be partitioned as in \eqref{col_partition} with
	\begin{align*}
	&\collectionS_1=\bigl\{ \{1\},\{2\},\{3\},\{4\},\{5\}\bigr\},\quad
	\collectionS_2 = \bigl\{ \{1,2\}, \{2,3\}, \{3,4\},\{4,5\} \bigr\},\nonumber\\
	&\collectionS_3 = \bigl\{ \{1,2,3\},\{3,4,5\} \bigr\},\qquad\quad
	\collectionS_4=\collectionS_5=\emptyset.
	\end{align*}
	Based on Definition \ref{def:nonoverlap}, we can verify that the maximum non-overlapping collection $\tilde{\collectionS}_1$ is ${\collectionS}_1$ itself. Hence,
		${\collectionU}_1=\tilde{\collectionS}_1=\collectionS_1$ and $D_1=1$. 
		Next, we put the $D_2=3$ possible
		maximum non-overlapping collections $\tilde{\collectionS}_2$ in ${\collectionU}_2$:
		\begin{equation*}
		{\collectionU}_2 =\biggl\{ \bigl\{ \{1,2\}, \{3,4\} \bigr\}, \quad   \bigl\{  \{2,3\}, \{4,5\} \bigr\}, \quad   \bigl\{ \{1,2\}, \{4,5\} \bigr\}\biggr\}.
		\end{equation*}%
		Similarly, there are $D_3=2$ possible maximum non-overlapping collections $\tilde{\collectionS}_3$ in ${\collectionU}_3$:
		\begin{equation*}
		{\collectionU}_3=\biggl\{  \bigl\{ \{1,2,3\}\bigr\}, \quad   \bigl\{ \{3,4,5\} \bigr\}\biggr\}.
		\end{equation*}%
		%
		As a result, there are $D_1 D_2 D_3 = 6$ possible maximum non-overlapping
		collections of $\collectionS$ in ${\collectionU}$: 
		\begin{equation}
		{\collectionU} = \left\{\bigcup_{k=1}^3
		\tilde{\collectionS}_k| \tilde{\collectionS}_k\in{\collectionU}_k, k=1,2,3\right\}.
		\end{equation}%
		Let us consider $\tilde{\collectionS} =
		\{\{1\},\{2\},\{3\},\{4\},\{5\}, \{1,2\}, \{3,4\}, \{1,2,3\}\}\in\collectionU$ without
		loss of generality. User~1 has to decode the streams in
		${\collectionK}^{(1)} \cap \tilde{\collectionS}= \{\{1\},
		\{1,2\}, \{1,2,3\}\}$. This can be done in order $\pmb{\tilde{\pi}}^{(1)}_{\tilde{\collectionS}}:\{1,2,3\}\to
		\{1,2\} \to \{1\}$ according to decreasing cardinality of the sets.
		Similarly, user~2 decodes in order $\pmb{\tilde{\pi}}^{(2)}_{\tilde{\collectionS}}:\{1,2,3\}\to \{1,2\} \to \{2\}$;
		user~3 decodes in order $\pmb{\tilde{\pi}}^{(3)}_{\tilde{\collectionS}}: \{1,2,3\}\to \{3,4\} \to \{3\}$; 
		user~4 decodes in order $\pmb{\tilde{\pi}}^{(4)}_{\tilde{\collectionS}}: \{3,4\} \to \{4\}$; user~5 decodes $\{5\}$. Note that for any of the six collections $\tilde{\collectionS}\in\collectionU$,
		one can apply the same procedure, and finally return the best collection  $\tilde{\collectionS}^{\diamond}$ with the highest rate $R_{\text{SR}}^{\diamond}$. 
\end{example}
\begin{remark}
	If the users decode the streams with increasing cardinality
        order instead, each user first decodes its own private message
        followed by the common messages. In such case, to mitigate the interference while decoding the private message, the power allocated to common messages is likely to be suppressed under noise level, and the benefit of the general RS is therefore not fully exploited. 
\end{remark}

The proposed stream selection algorithm essentially reduces
the number of active streams and removes the excessive search over the decoding order.
However, stream selection itself adds extra complexity. Thus, we would
like to investigate on the overall complexity of Algorithm \ref{algo:optimize_with_stream_selection}, i.e., solve
$\Pc_{\text{SR}}^{\text{SD}}$ with the proposed stream selection
algorithm. Algorithm \ref{algo:optimize_with_stream_selection} consists of three parts, 1) perform SEA, 2)
select different $\tilde{\collectionS}$ and assemble them in
$\tilde{\collectionU}$, 3) solve
$\Pc_{\text{SR}}^{\text{SD}}(\pmb{\tilde{\pi}}_{\tilde{\collectionS}})$
for each $\tilde{\collectionS}$. The
complexity of applying SEA to reduce the $2^K-1$ streams to at most
$N_{\text{SEA}}$ streams is $\mathcal O(KM^22^K)$ \cite{li2018linearly}.
Then to find all the possible $\tilde{\collectionS}$, we can first
establish a lookup table offline for a sufficiently large $K$. As long
as $N_{\text{SEA}}\le 2^K-1$, we can construct $\tilde{\collectionU}$ by
searching in the table, regardless of the channel realization.
Therefore, the complexity of finding different $\tilde{\collectionS}$
can be neglected. The complexity of the third step involves the complexity of solving
$\Pc_{\text{SR}}^{\text{SD}}(\pmb{\tilde{\pi}}_{\tilde{\collectionS}})$
for all $\tilde{\collectionS}$.  However, it is hard to characterize
the exact value of $D_k$ for an arbitrary $N_{\text{SEA}}$ in general. The
difficulty comes from the unknown overlap among the remaining streams
after SEA. Therefore, we present an upper bound on the complexity for
this part. For a given $\tilde{\collectionS}$,
$\Pc_{\text{SR}}^{\text{SD}}(\pmb{\tilde{\pi}}_{\tilde{\collectionS}})$
is a convex problem with at most $(M^2+1)N_{\text{SEA}}$ variables and
$1+(K+1)N_{\text{SEA}}$ constraints. Therefore, the worst-case complexity for solving
$\Pc_{\text{SR}}^{\text{SD}}(\pmb{\tilde{\pi}}_{\tilde{\collectionS}})$ for a given $\tilde{\collectionS}$ using CCCP
 is $\mathcal O\left(M^6N_{\text{SEA}}^{3.5}K^{0.5}\right)$. Next, to
characterize $D_k$, we consider all the streams of layer $k$, i.e.,
$\binom{K}{k}$ in total, if there exists at least one layer-$k$ stream
after SEA. Note that such calculation may incorporate several streams
already eliminated by SEA, and hence leads to an upper bound. Let
$k^*({N_{\text{SEA}}})$ denote the smallest $k$ such that
$\sum_{k'=1}^k\binom{K}{k'}\ge N_{\text{SEA}}$. Then  the number of
possible $\tilde{\collectionS}$ is no greater than 
$\prod_{k=2}^{k^*({N_{\text{SEA}}})}\frac{\binom{K}{k}\binom{K-k}{k}\cdots\binom{K-(\lfloor
\frac{K}{k}\rfloor-1)\times k}{k}}{\lfloor \frac{K}{k} \rfloor\times
\max{(1,\lfloor \frac{K}{k}\rfloor-1)}}$, which  can be further upper
bounded by $\binom{K}{K/2}^{\frac{N_{\text{SEA}}}{2}}$. As a result, the
worst-case complexity of solving $\Pc_{\text{SR}}^{\text{SD}}$ with the
proposed stream selection is $\mathcal
O\left(KM^22^K+M^6N_{\text{SEA}}^{3.5}K^{0.5}\binom{K}{K/2}^{\frac{N_{\text{SEA}}}{2}}\right)$.
The above calculation is omitted for brevity. Let us consider $M=K=4$, then the complexity of solving $\Pc_{\text{SR}}^{\text{SD}}$ without stream selection is about $10^{26}$, while the upper bound on the worst-case complexity of solving $\Pc_{\text{SR}}^{\text{SD}}$ with our stream selection is about $10^{8}$ if $N_{\text{SEA}}=6$. The above comparison suggests that the complexity of precoder design after our stream selection algorithm is considerably reduced. 
\subsection{Imperfect CSIT}
In practice, the CSIT is obtained by direct estimation or limited
feedback from the receivers. Therefore, the information may be
inaccurate or outdated. A common model for imperfect CSIT is to assume that 
\begin{equation}
  \rvMat{H} = \hat{\rvMat{H}} + \tilde{\rvMat{H}},\label{eq:imperfect_H}
\end{equation}%
where $\hat{\rvMat{H}}$ is the channel estimate and $\tilde{\rvMat{H}}$
is the estimation error that is unknown to the transmitter, and they are
independent. In addition, we assume for simplicity that $\tilde{\rvMat{H}}$ has
i.i.d.~entries with variance $\sigma^2$ each. 
 Let $\hat{\Hm}$ be a particular realization of $\hat{\rvMat{H}}$. The
precoder design should only depend on the estimate $\hat{\Hm}$ and the
CSIT inaccuracy. 

A proper way to characterize the performance with
imperfect CSIT is the outage formulation, e.g., to find out the achievable
rate tuples for a given outage probability. But it is hard to derive
closed-form expressions exploitable for precoder optimization. Therefore,
instead of reformulating entirely the problem with imperfect CSIT, we seek an
adaptation of the one derived for perfect CSIT. To that end, we adopt an
ergodic formulation which replaces \eqref{eq:region_JD} with the following constraint \footnote{Strictly speaking, the expected rate \eqref{eq:region_JD2} is achievable only when the channel estimate is fixed during the transmission and known at the transmitter side, while the channel estimation error is time varying according to a given distribution.}
\begin{equation}
\sum_{\mathcal{K}\in \collectionS^{(k)}} {R}_{\mathcal{K}} \le
\mathbb{E}_{\tilde{\rvVec{H}}_k}\left\{ \log\biggl(1 +
\frac{\sum_{\mathcal{K}\in \collectionS^{(k)}} \rvVec{H}_k^\H
\Qm_{\mathcal{K}} \rvVec{H}_k}{
	1 + \sum_{\mathcal{K}'\not\in {\collectionK}^{(k)}}
        \rvVec{H}_k^\H \Qm_{\mathcal{K}'} \rvVec{H}_k}
        \biggr) \right\},\quad\forall\, k\in[K], \collectionS^{(k)} \subseteq
\collectionK^{(k)},\label{eq:region_JD2}  
\end{equation}
where $\rvVec{H}_k := \hat{\hv}_k + \tilde{\rvVec{H}}_k$ is the random
channel vector; and the expectation is over the CSI error. Note that in the
above expression, we implicitly assume that the random channel vector is
known at the receivers, whereas only the expectation on
the right-hand side of \eqref{eq:region_JD2}
is known at the transmitter's side. Such formulation is also
widely used in the literature of robust optimization \cite{boyd2004convex, charnes1959chance}.
There are two issues with the
above formulation though. First, we need to know the exact distribution of the
CSI error to compute the expectation. This may not be possible in
practice. Second, even with the exact distribution, finding the
expectation can be computationally costly. In order to capture the CSI
error in a simple way, we work with a lower bound on the expectation on
the right-hand side of \eqref{eq:region_JD2}. 
\begin{claim} \label{claim:imperfect}If the CSI error is symmetrically distributed such that
  $-\tilde{\rvMat{H}}$ has the same distribution as
  $\tilde{\rvMat{H}}$, then we have the following lower bound on the expectation
  \begin{multline}
\mathbb{E}_{\tilde{\rvVec{H}}_k}\left\{ \log\biggl(1 +
\frac{\sum_{\mathcal{K}\in \collectionS^{(k)}} \rvVec{H}_k^\H
\Qm_{\mathcal{K}} \rvVec{H}_k}{
	1 + \sum_{\mathcal{K}'\not\in {\collectionK}^{(k)}}
        \rvVec{H}_k^\H \Qm_{\mathcal{K}'} \rvVec{H}_k}
        \biggr) \right\} \\ \ge 
 \log\biggl(1 +
 \frac{\sum_{\mathcal{K}\in \collectionS^{(k)}} \hat{\hv}_k^\H
 \Qm_{\mathcal{K}} \hat{\hv}_k}{
	1 + \sum_{\mathcal{K}'\not\in {\collectionK}^{(k)}}
        \bigl( \hat{\hv}_k^\H \Qm_{\mathcal{K}'} \hat{\hv}_k + \sigma^2
        \trace( \Qm_{\mathcal{K}'} ) \bigr) }  
        \biggr),  \quad\forall\, k\in[K], \collectionS^{(k)} \subseteq
\collectionK^{(k)}. 
   \end{multline}%
\end{claim}
\begin{proof}
  The expectation can be rewritten as the difference of two
  expectations, i.e.,
  \begin{multline}
\mathbb{E}_{\tilde{\rvVec{H}}_k}\left\{ \log\biggl(1 +
\frac{\sum_{\mathcal{K}\in \collectionS^{(k)}} \rvVec{H}_k^\H
\Qm_{\mathcal{K}} \rvVec{H}_k}{ 1 + \sum_{\mathcal{K}'\not\in
{\collectionK}^{(k)}} \rvVec{H}_k^\H \Qm_{\mathcal{K}'} \rvVec{H}_k}
\biggr) \right\} \\ =  
\mathbb{E}_{\tilde{\rvVec{H}}_k}\!\!\left\{ \log\left(1 + \rvVec{H}_k^\H
 \left( \sum_{\mathcal{K}'\not\in {\collectionK}^{(k)}} 
\!\!\Qm_{\mathcal{K}'}  +\!\!\!\! 
\sum_{\mathcal{K}\in \collectionS^{(k)}} 
\!\!\Qm_{\mathcal{K}} \right) \right) \!\rvVec{H}_k \!\right\}
- 
\mathbb{E}_{\tilde{\rvVec{H}}_k}\!\!\left\{
 \log\left( 1 +\!\! \sum_{\mathcal{K}'\not\in {\collectionK}^{(k)}}
 \rvVec{H}_k^\H \Qm_{\mathcal{K}'} \rvVec{H}_k \right) \right\}, \label{eq:tmp332}
      \end{multline}
where the second term can be upper bounded with Jensen's inequality as 
\begin{equation}
  \mathbb{E}_{\tilde{\rvVec{H}}_k}\left\{
 \log\left( 1 + \sum_{\mathcal{K}'\not\in {\collectionK}^{(k)}}
 \rvVec{H}_k^\H \Qm_{\mathcal{K}'} \rvVec{H}_k \right) \right\} \le 
\log\biggl(
	1 + \sum_{\mathcal{K}'\not\in {\collectionK}^{(k)}}
        \bigl( \hat{\hv}_k^\H \Qm_{\mathcal{K}'} \hat{\hv}_k + \sigma^2
        \trace( \Qm_{\mathcal{K}'} ) \bigr) 
        \biggr). 
\end{equation}%
Next, we look at the first term on the right-hand side of
\eqref{eq:tmp332}. We consider a simple form of it as
\begin{equation}
  f(\sigma) := \mathbb{E}_{\rvVec{G}}\left\{ \log (1+ (\hat{\hv} +
  \sigma \rvVec{G})^\H \Am (\hat{\hv} + \sigma \rvVec{G})) \right\},
\end{equation}%
where the random vector $\sigma \rvVec{G}$ has the same distribution as
$\tilde{\rvVec{H}}_k$; $\Am$ is some positive semi-definite matrix. 
To prove the claim, it is enough to show that
$f(\sigma)$ defined above is non-decreasing with $\sigma$. To that end,
let us take the first derivative of $f(\sigma)$. 
\begin{align}
  \frac{\mathrm{d} f(\sigma)}{\mathrm{d} \sigma} &= \log e\cdot
  \mathbb{E}_{\rvVec{G}}\left\{ \frac{ 2\sigma \rvVec{G}^\H \Am
  \rvVec{G} + 2   \mathrm{Re}\{\rvVec{G}^\H \Am \hat{\hv}\}}{1+ (\hat{\hv} +
  \sigma \rvVec{G})^\H \Am (\hat{\hv} + \sigma \rvVec{G})}  \right\}
  \label{eq:tmp1} \\
  &\ge \log e\cdot \mathbb{E}_{\rvVec{G}}\left\{ \frac{  2  \mathrm{Re}\{\rvVec{G}^\H \Am \hat{\hv}\}}{1+ (\hat{\hv} + \sigma \rvVec{G})^\H \Am (\hat{\hv} + \sigma \rvVec{G})}  \right\} \label{eq:tmp2}\\
  &\ge \log e\cdot \mathbb{E}_{\rvVec{G}}\left\{ \frac{  2
  \mathrm{Re}\{\rvVec{G}^\H \Am \hat{\hv}\}}{1+ \max_{b\in\{1,-1\}}
  \left\{
  (\hat{\hv} + b \sigma \rvVec{G})^\H \Am (\hat{\hv} + b \sigma
  \rvVec{G}) \right\}}  \right\} \label{eq:tmp3}\\
  &= 0 \label{eq:tmp4}
\end{align}%
where \eqref{eq:tmp1} is obtained by putting the derivative inside the
expectation;\footnote{We can easily check that this can be done, e.g.,
using the dominated convergence theorem.} \eqref{eq:tmp2} comes from the
fact that $\Am$ is positive semi-definite, i.e., $\rvVec{G}^\H \Am
\rvVec{G} \ge 0$ with probability $1$; we introduced $b\in\{1,-1\}$ to
obtain the lower bound \eqref{eq:tmp3} and to symmetrize the denominator
with respect to $\rvVec{G}$; the last equality is from the fact that
$\rvVec{G}$ is symmetrically distributed and that the function inside the expectation is odd
with respect to $\rvVec{G}$. Indeed, for any odd function $g(\cdot)$, we
have $\mathbb{E}\left\{ g(\rvVec{G}) \right\} = \mathbb{E}\left\{
g(-\rvVec{G}) \right\} = -\mathbb{E}\left\{ g(\rvVec{G}) \right\}$ due
to the symmetry of the distribution of $\rvVec{G}$ and the symmetry of
the function, respectively. This leads to  $\mathbb{E}\left\{
g(\rvVec{G}) \right\} = 0$. 
\end{proof}
Note that the symmetry of the CSI error is a mild assumption that can be
satisfied in practical situations. Therefore, to take into account the
CSI inaccuracy for joint decoding, we propose to replace \eqref{eq:region_JD} with the following constraint
\begin{equation}
\sum_{\mathcal{K}\in \collectionS^{(k)}} {R}_{\mathcal{K}} \le
\log\biggl(1 +
 \frac{\sum_{\mathcal{K}\in \collectionS^{(k)}} \hat{\hv}_k^\H
 \Qm_{\mathcal{K}} \hat{\hv}_k}{
	1 + \sum_{\mathcal{K}'\not\in {\collectionK}^{(k)}}
        \bigl( \hat{\hv}_k^\H \Qm_{\mathcal{K}'} \hat{\hv}_k + \sigma^2
        \trace( \Qm_{\mathcal{K}'} ) \bigr) }  
        \biggr),  \quad\forall\, k\in[K], \collectionS^{(k)} \subseteq
\collectionK^{(k)}.\label{eq:region_JD3}  
\end{equation}%
Similarly, for successive decoding, we can replace \eqref{region_SIC} with the following constraint
\begin{align}
&{R}_{\Kc^{(k)}_{\pi_{k,n}}} \le\nonumber\\
&\log\biggl(1+\frac{\hat{\hv}_k^H\Qm_{\Kc^{(k)}_{\pi_{k,n}}}\hat{\hv}_k}{1+\sum_{{\Kc'}\not\in
		\collectionK^{(k)}}\bigl( \hat{\hv}_k^H\Qm_{{\Kc'}}\hat{\hv}_k+ \sigma^2
	\trace( \Qm_{\mathcal{K}'} )\bigr)+\sum_{n'>n}\bigl(\hat{\hv}_k^H\Qm_{\Kc^{(k)}_{\pi_{k,n'}}}\hat{\hv}_k+ \sigma^2
	\trace( \Qm_{\Kc^{(k)}_{\pi_{k,n'}}} )\bigr)}\biggr),\nonumber\\
&\qquad\qquad \qquad\qquad \qquad\qquad \qquad\qquad \qquad\qquad \forall\,k\in[K], n\in[2^{K-1}].\label{imperfect_SIC}
\end{align}
With the above constraints, we can still apply CCCP to
solve the precoder design problems under imperfect CSIT, as presented in the Section \ref{sec:op_JD} and in Section \ref{sec:op_SIC}, since the
presences of the CSI error terms in \eqref{eq:region_JD3} and \eqref{imperfect_SIC} do not change
the convexity of the denominators. 
Essentially, here we try to optimize some lower bounds of the rate
functions over a set of covariance matrices. Then we can apply the
precoders returned by the algorithms, and achieve at least as good
as the lower bounds predict.

\section{Numerical Results}\label{sec:simu}
In this section, we provide some numerical results to illustrate
the performance of the proposed algorithms. While we focus on the
sum rate performance in this section, similar conclusions
can be obtained for the weighted sum rate, worst-user rate, or minimum power
scenarios. Assuming the spatially correlated Rayleigh-fading channel,
we have $\hv_k\sim\mathcal{CN}(\bm{0},\Rm_k)$, where $\Rm_k$ is a
positive semi-definite channel covariance matrix. With the
Karhunen-Loeve representation, the downlink channel $\hv_k$ of user $k$
is in the form $\hv_k=\Um_k \bm{\Lambda}_k^{\frac{1}{2}}\wv_k$, where
$\wv_k\in \mathbb{C}^{r_k\times1}\sim\mathcal{CN}(\bm{0},\Id)$,
$\bm{\Lambda}_k$ is an $r_k\times r_k$ diagonal matrix whose elements
are the nonzero eigenvalues of $\Rm_k$, and $\Um_k\in
\mathbb{C}^{M\times r_k}$ is the tall unitary matrix formed by the
corresponding eigenvectors.

We further consider the one-ring scattering model \cite{adhikary2013joint}. Then the correlation between the channel coefficients of antennas $1\le m,p \le M$ is given by 
\begin{align}
[\Rm_k]_{m,p}=\frac{1}{2\Delta_k}\int_{{\theta _k} - {\Delta_k}}^{{\theta_k} + {\Delta_k}} {{e^{ - j\frac{2\pi}{\lambda}\bm{\Phi}(\alpha) (\uv_m-\uv_p)}}d\alpha } ,
\end{align}
where $\theta_k$ is the azimuth angle of user $k$ with respect to the orientation perpendicular to the array axis, $\Delta_k$ indicates the angular spread (AS) of departure to user $k$, $\bm{\Phi}(\alpha)=[\cos(\alpha),\sin(\alpha)]$ is the wave vector for a planar wave impinging with the angle of $\alpha$, $\lambda$ is the wavelength and $\uv_m=[x_m,y_m]^{\T}$ is the vector indicating the position of BS antenna $m$ in the 2-D coordinate system. Let us consider the special but important case of a uniform linear array (ULA) placed at the origin along the y-axis. Denoting by $\lambda D$ the spacing of antenna elements, the covariance matrix of the channel is given by the Toeplitz form 
\begin{align} 
	[\Rm_k]_{m,p}=\frac{1}{2\Delta_k}\int_{{\theta _k} - {\Delta _k}}^{{\theta _k} + {\Delta _k}} {{e^{ - j2\pi D(m - p)\sin (\alpha )}}d\alpha }.\label{eq:one-ring}
\end{align}

Suppose that $K$ users are selected to form $G$ groups based on the
similarity of their channel covariance matrices. We let $K_g$ denote the
number of users in group $g$, such that $K=\sum\nolimits_{g=1}^GK_g$. We
make the same assumption as in \cite{adhikary2013joint} that users in the
same group $g$ share the same covariance matrix $\Rm_g=\Um_g\bm{\Lambda}_g\Um_g$. More precisely, the channel vector of user $k$ in group $g$ is given by $\hv_k=\Um_g \bm{\Lambda}_g^{\frac{1}{2}}\wv_k$.

The baseline schemes are the following:
\begin{itemize}
	\item The sum capacity $C_{\text{sum}}$ that can be derived with the
	MAC-BC duality as 
	\begin{align}
		C_{\text{sum}}=\max_{\sum_i P_i\le
			P}\log\det(\Id+\Hm^H\diag(P_1,\ldots,P_K)\Hm). 
	\end{align}%
        The above problem is convex and can be solved efficiently.
      \item The unicast scheme with the covariance matrices obtained by solving the following problem 
		\begin{align*}
	\Pc_{\text{uni}}:\quad&\mathop {\max }\limits_{\Qm_k\succeq 0}\sum_{k\in[K]}R_k\nonumber\\
	&\text{s.t. } \;\eqref{eq:unicast_SINR},\sum_{k\in [K]} \trace\left(\Qm_k\right) \le P.
	\end{align*}
      \item The 1-layer RS (the only common message sent by the BS is intended to all the
	users) with successive decoding \cite{clerckx2019rate, joudeh2016robust,yang2019optimization}.
      \item The ZF scheme with precoder $(\Hm^\H)^\dag \diag\{\sqrt{\pv^{\text{opt}}}\}$ \cite{wiesel2008zero}, where $\pv^{\text{opt}}$ is the solution to
	\begin{align}
	\max_{\pv\ge 0}  \sum_k \log(1+p_k)\;\;\text{s.t.  } \sum\nolimits_k p_k[(\Hm^\H \Hm)^{-1}]_{k,k}\le P.
	\end{align}
\end{itemize}
 We emphasize that, for comparison fairness, we apply a similar CCCP to
 solve the precoder optimizations for both the unicast scheme and
 the 1-layer RS scheme. Note that the algorithms proposed in the
 references for the 1-layer RS do not necessarily consider power
 optimization for each individual message, resulting a strictly
 suboptimal solution compared to that obtained with CCCP. 

We first verify the performance of the general RS scheme with joint
decoding for a small number of users, i.e., $K=3$. To that end, we
assume $G=2$ groups, and each user can be inside any of the clusters
with equal probability. We set the azimuth angle of the $g$-th group as
$\theta_g=-\frac{\pi}{3}+\frac{15\pi}{180}(g-1)$ and assume the same AS
for all groups as $\Delta=\frac{5\pi}{180}$. Algorithm \ref{algo:max_sr}
is applied to obtain a stationary point of the sum rate of the general
RS with joint decoding. It is demonstrated in Fig. \ref{fig:MM_vs_ZF}
that both the 1-layer RS and the general RS considerably increase
the sum rate compared to unicast, especially at high SNR. In addition,
with only 3 more streams in the general RS, we can improve the sum rate
by more than 1 bit compared to the 1-layer RS at medium-high SNR.
As the users in the same group have correlated channels, it is
reasonable to decode the interference inside the group instead of
treating it as noise. Since the users are spatially correlated in
the one-ring scattering model, the performance of the ZF scheme
degrades severely as demonstrated in Fig. \ref{fig:MM_vs_ZF}. Therefore, in
the following, we will not consider ZF.
\begin{figure}
	\centering
	\includegraphics[width=0.55\textwidth]{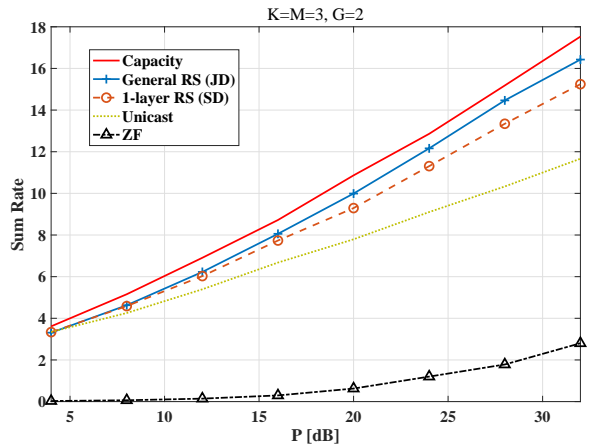}
	\caption{Comparison of the capacity, sum rates under the general RS (JD), 1-layer RS (SD), unicast and ZF, $K=3, M=3, G=2,\Delta=\frac{5\pi}{180}, \theta_g=-\frac{\pi}{3}+\frac{15\pi}{180}(g-1) $.}\label{fig:MM_vs_ZF}
\end{figure}

 Now let us adopt Algorithm
\ref{algo:optimize_with_stream_selection} to solve
$\Pc_{\text{SR}}^{\text{SD}}$ with the proposed stream selection for the
4-user case. We set $N_{\text{SEA}}=15$. Thus the first step in
Algorithm \ref{algo:stream_selection} actually preserves all the $2^K-1$
steams. Since each $\tilde{\collectionS}$ generated by Algorithm
\ref{algo:stream_selection} contains $8$ streams instead of the original
$15$ streams, the decoding complexity for each user is substantially
reduced. Recall that for a given collection $\tilde{\collectionS}$, the
decoding order is fixed after the stream selection as explained in
Section \ref{sec:stream_sel}. We consider two scenarios with disjoint
and overlapping eigen-subspaces, respectively, in Fig.
\ref{fig:4_user_SIC}. We set
$\theta_g=-\frac{\pi}{3}+\Delta+\frac{15\pi}{180}(g-1)$ and
$\Delta=\frac{5\pi}{180}$ for the former scenario, while $\theta_g=-\frac{\pi}{3}+\frac{10\pi}{180}(g-1)$
and $\Delta=\frac{15\pi}{180}$ for the latter. In the scenario
with disjoint eigen-subspaces, inter-group interference is small and
common messages across different groups are not needed,
the 1-layer RS boils down to the unicast scheme. It is
worth mentioning that in the simulation, the 1-layer RS slightly
outperforms the unicast since we randomly put four users in two groups. Therefore, it is possible that all the users are in the same group so that they all benefit from the common message. The general RS after stream selection, which has three additional streams, further improves the sum rate. However, in such disjoint case, extra streams only brings slight intra-group rate increase. In contrast, in the overlapping case, the common message to all the users increases the sum rate by 2 bits compared to unicast. Apart from this, a large gain is further enabled by three additional inter-group common messages in the general RS.
\begin{figure}
	\centering
	\includegraphics[width=0.8\textwidth]{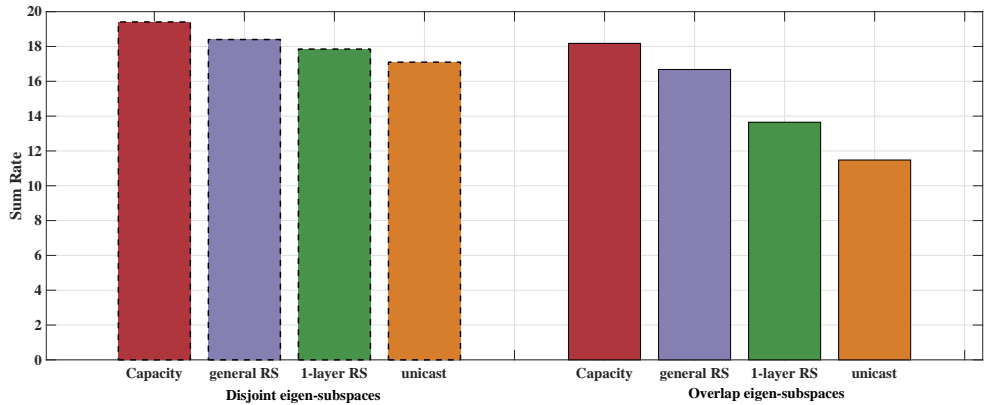}
	\caption{Comparison of the capacity, the sum rates under the
        general RS after stream selection (SD), the 1-layer RS scheme~(SD), and
        the unicast scheme. $K=4, M=4, G=2, P=30\,$dB. Two considered
        scenarios:~1) disjoint eigen-spaces~($\Delta=\frac{5\pi}{180},
        \theta_g=-\frac{\pi}{3}+\Delta+\frac{15\pi}{180}(g-1)$),
        2)~overlapping eigen-spaces~($\Delta=\frac{15\pi}{180},
        \theta_g=-\frac{\pi}{3}+\frac{10\pi}{180}(g-1)$).}\label{fig:4_user_SIC}
\end{figure} 

In Fig. \ref{fig:compare_large_K}, we consider a larger number of users,
i.e., $K=8$. In this simulation, we assume that 8 users can be inside
any of the two groups with equal probability. We set
$\Delta=\frac{20\pi}{180}$ and
$\theta_g=-\frac{\pi}{3}+\frac{\pi}{8}(g-1)$, which corresponds to
overlapping eigen-spaces between the two groups. In this case, the
general RS scheme involves 255 streams, which is infeasible due to
high complexity. Thanks to the stream selection algorithm, it is
possible to investigate the performance of the general RS under such
relatively large number of users. We set
$N_{\text{SEA}}=38$, then each $\tilde{\collectionS}$ generated by
Algorithm \ref{algo:stream_selection} contains 30 streams. Furthermore,
the number of constraints in successive decoding is reduced to 16, which
is implementable in practice. Fig. \ref{fig:compare_large_K} suggests that the 1-layer RS slightly
outperforms the unicast, while the general RS further provides a large gain thanks to the use of
multiple inter-group common messages. This observation
confirms the effectiveness of our stream selection in identifying 
``good'' streams. 
\begin{figure}
	\centering
	\includegraphics[width=0.6\textwidth]{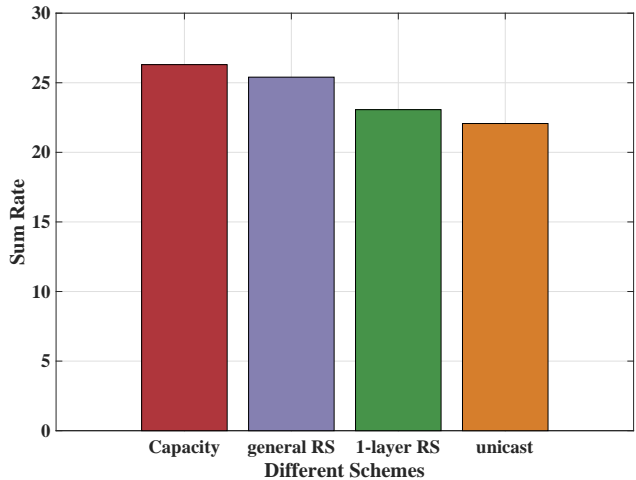}
	\caption{Comparison of the capacity, the sum rates under the
        general RS after stream selection (SD), the 1-layer RS
        scheme~(SD), and the unicast scheme. $K=M=8, P=30 $dB, $G=2, \Delta=\frac{20\pi}{180}, \theta_g=-\frac{\pi}{3}+\frac{\pi}{8}(g-1)$.}\label{fig:compare_large_K}
\end{figure}

Finally, let us consider i.i.d.~channels and investigate the performance of the general RS with imperfect CSIT. 
We assume that $\rvMat{H} = \hat{\rvMat{H}} + \tilde{\rvMat{H}}$ as in
\eqref{eq:imperfect_H}. Here, we let the entries of $\tilde{\rvMat{H}}$
and $\hat{\rvMat{H}}$ be i.i.d.~circularly symmetric Gaussian with
variances $\sigma^2$ and $1-\sigma^2$, respectively. 
In Fig. \ref{fig:imperfect_gain}, the sum rate with the proposed
regularization is calculated as follows. We first optimize the sum rate
based on the lower bound, i.e., replace the rate constraints in
\eqref{region_SIC} by those in \eqref{imperfect_SIC} in Algorithm
\ref{algo:max_sr_SIC}, and achieve a group of feasible $\Qm$. Then, we
substitute such returned $\Qm$ into the real rate constraints in
\eqref{region_SIC} to get a set of achievable rates, and thus the sum
rate. To obtain the sum rate of the general RS without regularization,
we solve $\Pc_{\text{SR}}^{\text{SIC}}$ with $\hat{\Hm}$ directly
replacing $\Hm$, and then substitute the returned $\Qm$ into \eqref{region_SIC}
with the real channel matrix $\Hm$ to get the achievable rates. The
results show that the proposed regularization brings remarkable rate
improvement especially when the CSIT error $\sigma^2$ is large, e.g., almost
up to $200\%$ gain when $\sigma^2=0.9$. 
In Fig.\ref{fig:imperfect_rate}, we compare the performance of the
unicast and the general RS after stream selection. The sum rate of the
unicast scheme with regularization is obtained by steps similar to  those for
the general RS with regularization. We can observe that the general RS
always outperforms the unicast scheme even with a large estimation error. 

\begin{figure}
	\begin{minipage}{0.48\textwidth}
		\centering
		\includegraphics[width=1.1\textwidth]{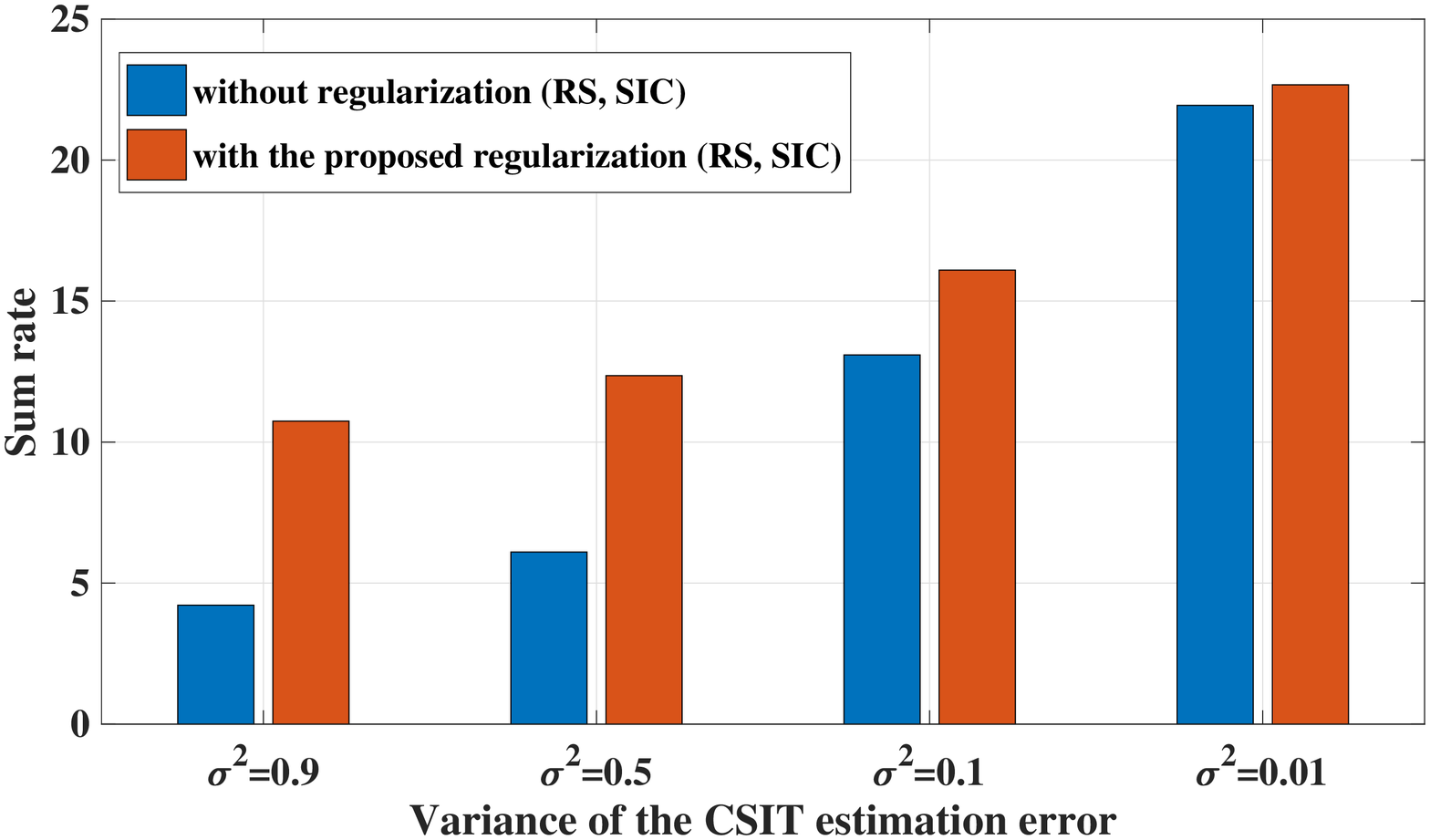}
		\caption{Sum rate comparisons without and with the proposed regularization, imperfect CSIT, $K=M=4, P=30 $dB, $N_{\text{SEA}}=2^K-1$, both with the general RS, i.i.d. channels.}\label{fig:imperfect_gain}
	\end{minipage}
	\hfill
	\begin{minipage}{0.48\textwidth}
			\centering
		\includegraphics[width=1.1\textwidth]{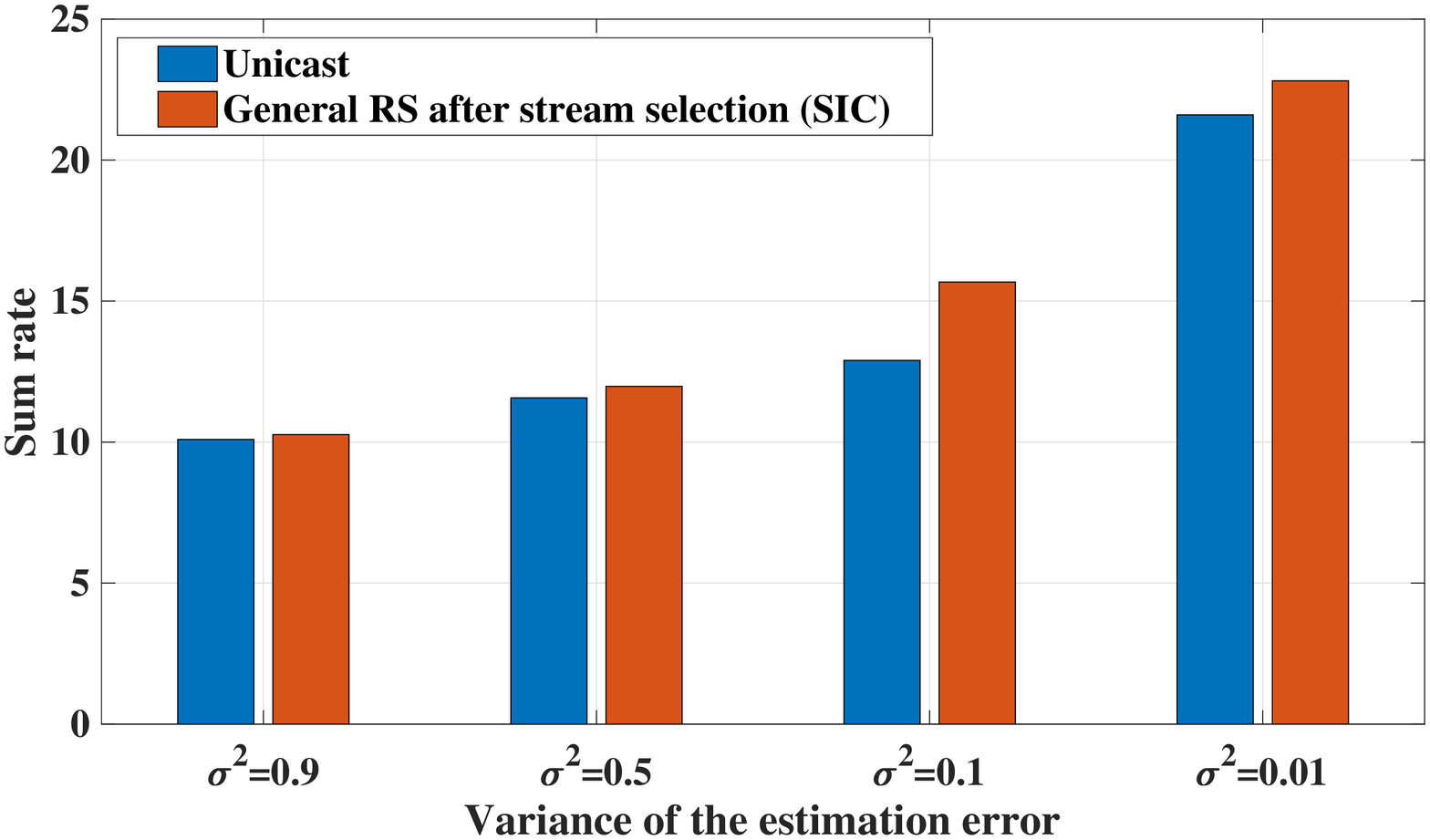}
		\caption{Sum rate comparisons under the general RS and unicast, both with the proposed regularization, imperfect CSIT, $K=M=4, P=30 $dB, $N_{\text{SEA}}=2^K-1$, i.i.d. channels.}\label{fig:imperfect_rate}
	\end{minipage}
\end{figure}

\section{Conclusion}\label{sec:con}
In this work, we have investigated the  general RS
scheme applied for multi-antenna downlink communications. We have
proposed a full range of novel solutions including precoder design,
stream selection, and imperfect CSIT regularization. We have run
numerical simulations showing that our RS solutions, even under
practical constraints, can provide substantial performance gains over
existing schemes. It is worth noting that since RS is linear, the implementation cost of the proposed algorithms are comparable to those applied in practical systems. In summary, our study has demonstrated that the general RS is a viable way to mitigate interference in future multi-antenna networks.
\bibliographystyle{IEEEtran}
\bibliography{IEEEabrv,RS_op2019}
\end{document}